\pdfoutput=1
\documentclass[11pt]{article}
\usepackage[T1]{fontenc}
\usepackage[utf8]{inputenc}
\IfFileExists{lmodern.sty}{\usepackage{lmodern}}{}
\usepackage{amsmath,amssymb,amsthm}
\usepackage[a4paper,margin=2.9cm]{geometry}
\usepackage{graphicx}
\usepackage{xcolor}
\IfFileExists{lmodern.sty}{\usepackage{microtype}}{}
\usepackage[font=small,labelfont=bf]{caption}
\usepackage{tikz}
\usetikzlibrary{arrows.meta,positioning}
\usepackage[hidelinks]{hyperref}

\numberwithin{equation}{section}
\newtheorem{proposition}{Proposition}[section]
\newtheorem{lemma}{Lemma}[section]
\theoremstyle{definition}
\newtheorem{fact}{Fact}
\newtheorem{definition}{Definition}[section]
\newtheorem{remark}{Remark}[section]

\newenvironment{keywords}{\par\medskip\noindent\textbf{Key words.}\ }{\par}
\newenvironment{AMS}{\par\medskip\noindent\textbf{MSC codes.}\ }{\par\medskip}

\newcommand{\E}{\mathbb{E}}
\newcommand{\daw}{\operatorname{D}}
\newcommand{\erfi}{\operatorname{erfi}}
\newcommand{\sG}{s_{G}}

\title{Optimal Trading of Microstructure Mean Reversion}
\author{Lucas Rabechini Amaral\thanks{Preliminary draft; comments
welcome. E-mail: \texttt{lucasra@gmail.com}. The views expressed are
solely the author's and do not reflect those of any institution with
which he is or has been affiliated; the paper is theoretical, uses no
proprietary or confidential information, and nothing herein constitutes
investment advice.}}
\date{This version: July 2026}

\begin{document}
\maketitle

\begin{abstract}
At the scale of seconds the observed mid carries a stationary,
mean-reverting error around a latent efficient price. We build an order
book whose own flow produces that error, and solve for the trading rule
that maximises the long-run average profit rate, net of the bid--ask
spread. In a liquid large-tick asset the queue at the touch is deep. The
spread is therefore one tick or two, and we assume our one-lot trader
does not enter the book's intensities. The mid is the bid plus half the
spread. It moves on the half-tick grid, and the spread is exactly its
parity: tight at one tick when the mid, measured in ticks, sits at a
half-integer; open at two when it sits at an integer. One coordinate therefore carries the problem: the
gap $G$ between mid and efficient price. We take that price to be an
exogenous Brownian martingale, and $G$ to be observable. The mid is a pure jump
process whose move intensities lean toward the efficient price. Under
one balanced-response condition, which equalises the book's corrective
drift across parities, mean reversion of $G$ is a theorem: its
conditional mean and stationary covariance are exactly those of an
Ornstein--Uhlenbeck process of reversion rate $\alpha$ and stationary
standard deviation $\sG$ (the pair the gap's autocovariance
identifies). Its
paths are not: the mid jumps. Passage times are therefore evaluated on
the Gaussian diffusion those two moments define, at an error we bound on
the reward side and leave heuristic on the timing side. A symmetric band
of half-width $\theta$ buys when the gap reaches $-\theta$, sells when it
reaches $+\theta$, and holds inside. On the surrogate such a band is
optimal among all admissible strategies, by results from the switching
literature; on the jump process itself that reduction remains a
conjecture. Write $\phi$ for the tight-book half-spread. Optimising over $\theta$
gives, to leading order in $\sG/\theta$, the optimal half-width $\theta^{*}$ and the
profit rate $R^{*}$ it earns:
\[
\theta^{*}\big(\theta^{*}-\phi\big)=\sG^{2}\,,
\qquad
R^{*}=\alpha\,\sG\sqrt{2/\pi}\;e^{-\theta^{*2}/2\sG^{2}}\,.
\]
Threshold times margin equals the stationary variance of the gap. Trading
as soon as the gap covers the spread earns zero: all profit is the option
value of waiting.
\end{abstract}




\section{Introduction}\label{sec:intro}

Optimal trading rests on two questions. The first is empirical: how do
the mid, the spread, and the order flow actually move at the time scale
where trading decisions live? The
second is the control problem: given those dynamics, what is the best
strategy, and what does it earn? This paper answers the second question
inside a model built, piece by piece, from the known answers to the
first.

At the intraday scale, four dynamics are documented across four
decades of microstructure evidence. (i) \emph{The mid's moves revert}:
successive moves are negatively autocorrelated at the finest time
scales \cite{hansenlunde06}, the fact that the mutually exciting price
models of Bacry et al. \cite{bacry13a,bacry13b}, built on Hawkes
processes \cite{hawkes71}, are designed to reproduce; we used the same
family for the clustering of large-order flow \cite{rabechini19}.
(ii) \emph{Displacements are transient}: a burst of orders pushes the
mid away from value, and part of the push reverts as the book
replenishes \cite{obizhaevawang13,bouchaud09}.
(iii) \emph{In liquid large-tick assets the spread is pinned}: when
the tick is not small relative to the volatility, the spread sits at
one tick almost all of the time and opens to two for brief episodes
\cite{dayrirosenbaum15,contdelarrard13}.
(iv) \emph{Beneath the noise, a diffusion}: strip the trading noise
statistically and what remains behaves as a Brownian semimartingale.
The realised-volatility literature models the intraday efficient price
this way and validates it on high-frequency data
\cite{andersen03,zhangmykland05}.
Together the four force the classical reading, from Roll \cite{roll84}
to Hasbrouck \cite{hasbrouck95}: the observed price is the diffusion of
(iv) plus a stationary error carrying (i) and (ii). The natural state variable is the \emph{gap}: the distance
between the observed mid and that latent efficient price.

The theory of optimal trading grew around these facts in four threads.
(1) Optimal execution (Almgren and Chriss \cite{almgren01}) schedules
a large order against the trade-off between impact and risk; (2) optimal
market making (Avellaneda and Stoikov \cite{as08}) quotes around
inventory risk, solved explicitly by Gu\'eant, Lehalle and
Fernandez-Tapia \cite{glf13} and developed into a toolkit by Cartea,
Jaimungal and Penalva \cite{cjp15} and Gu\'eant \cite{gueant16}. Guilbaud and
Pham \cite{guilbaudpham13} bring the tick grid into this thread (the
tick-valued spread a Markov chain, the trading in limit and market
orders), yet the mid there follows its own exogenous dynamics, pulled by
nothing. Throughout this thread, the
book is the venue but the price it quotes is modelled as a martingale:
the dynamics above enter as friction, not as object. (3) Another thread
trades the reversion itself: an Ornstein--Uhlenbeck spread, signal, or
reference price, taken as given. Threshold
rules are optimal there \cite{zhangzhang08,zervos13,leungli16}; Bertram
\cite{bertram10} maximises the long-run rate of the trading cycle
through first-passage times, the criterion we adopt; Lipton and
L\'opez de Prado \cite{liptonlopez20} give closed forms for a single
round trip; and, nearest on the making side, Ahuja, Papanicolaou, Ren
and Yang \cite{ahuja17} quote around a mean-reverting reference price.
That thread supplies the control theory we use, threshold
optimality and the renewal criterion, and takes the reverting object as
given. (4) The closest line builds the book itself around a
latent efficient price: moves fire at the boundaries of uncertainty
zones \cite{robertrosenbaum11}, the order flow estimates the price
\cite{delattre13}, the queue-reactive book re-centres its intensities
around a reference one \cite{huanglehalle15}. Sfendourakis
\cite{sfendourakis25} unifies the family and stabilises the mid around
the efficient price; Pulido, Rosenbaum and Sfendourakis \cite{pulido26}
solve full-information market making inside uncertainty zones. There
the book is the object of study, and the trading solved on it is market
making, around inventory risk.

This paper works the seam between threads (3) and (4): a book of the
fourth kind, on which the trading problem of the third is solved. Facts
(iii) and (iv) are the inputs:
we model precisely those liquid assets, and we take the efficient price
to be a Brownian martingale, the driftless core of the semimartingale
the econometrics finds. Inside that frame
we build an order book whose own dynamics produce facts (i)
and (ii), and we derive
the optimal trading band in closed form. The organising idea is a \emph{parity lock}: built from
the bid, the two-valued spread of (iii) is the parity, the
even-or-odd, of the mid on the half-tick grid, so gap and spread
collapse to one continuous coordinate plus a parity bit
(Section~\ref{sec:model}).

On that coordinate the mid is a pure jump
process whose move intensities lean toward the efficient price; under
one balanced-response condition on the flow, the gap's reversion, at the
book's leaning rate $\alpha$, is then a theorem:
its conditional mean and stationary covariance are exactly those of an
Ornstein--Uhlenbeck process. Its paths are not: the mid jumps.
The efficient price, in turn, is the long-run forecast of the mid,
Hasbrouck's permanent component (Section~\ref{sec:model}). The
strategy follows by renewal--reward, with waiting times evaluated on
the Gaussian surrogate those exact moments define, the model's central
approximation. It costs three errors, $\delta$ being the tick: (a) a
reward error of proved relative order $\delta/(\theta-\phi)$; (b) a
frozen tight-book cost, bounded by the occupancy of open books; (c) a
timing error of order $\delta/\theta$, which we leave heuristic. In the
regime $\delta\ll\theta-\phi$ where we work, (a) and (c) are small.

What comes out is a symmetric no-churn band at the optimal
half-width $\theta^{*}$: buy when the gap reaches $-\theta^{*}$, sell at
$+\theta^{*}$, hold inside. That half-width solves
$\theta^{*}(\theta^{*}-\phi)=\sG^{2}$ (threshold times margin equals
the gap's stationary variance) and the band earns rate
$\alpha\sG\sqrt{2/\pi}\,e^{-\theta^{*2}/2\sG^{2}}$. On the surrogate this
band is optimal among all admissible strategies, by results from the
switching literature; on the jump process itself that reduction remains
a conjecture (Section~\ref{sec:band}). The myopic benchmark (trade as
soon as the gap covers the spread) earns exactly zero on the
surrogate: the whole rate is the option value of waiting.

The band
equation is the large-threshold asymptote of an exact first-order
condition in the Dawson function, solved in Appendix~\ref{app:passage},
at whose root the rate formula is exact. Simulation places the best band
about a fifth inside that root; trading at the root costs three to four
percent of rate, at $\theta^{*}$ five to six (Section~\ref{sec:validate}).
Section~\ref{sec:remarks} collects what the solution teaches; proofs
and closed forms are in the appendix.

\section{The model}\label{sec:model}

\subsection{The phenomenon}\label{sec:evidence}

Picture the two prices of the same asset. One is the \emph{efficient
price} $X$: the fundamental value, real but not displayed anywhere.
It moves only on genuine news: on a filtered probability space
$(\Omega,\mathcal F,(\mathcal F_t)_{t\ge0},\mathbb P)$ we take it to be
an exogenous Brownian martingale,
formalised in Definition~\ref{def:X} below. This is not a convenience
but fact (iv) of the introduction, taken at its driftless
core \cite{andersen03,zhangmykland05}; and
Hasbrouck's permanent component is a random walk by
construction \cite{hasbrouck95}. The
other is the \emph{mid} $M$: the average of the best
bid and best ask actually displayed on the screen, observable, and moved
only by the book. Their difference is the \emph{gap}
\[
G_t := M_t - X_t ,
\]
a stationary error pulled back toward zero \cite{roll84,hansenlunde06}. The premise here (a modelling stance, not a law of nature) is the
existence of the latent martingale $X$; that the gap it defines is
stationary and mean-reverting is a finding in the evidence and a
theorem in the model below. Two properties are
worth recording now. First, the mid converges to $X$ \emph{in
conditional expectation}: the long-horizon forecast of the mid is
today's efficient price,
\begin{equation}\label{eq:X}
\lim_{h\to\infty}\E\big[M_{t+h}\mid\mathcal F_t\big]=X_t .
\end{equation}
Read left to right, \eqref{eq:X} is the pull: the forecast of where the
mid is going lands on value. Read right to left, it is an
identification: $X$ is
the permanent component of the observed price in the sense of Hasbrouck
\cite{hasbrouck95} (stated here, proved as a one-line theorem of the
model in Section~\ref{sec:derived}), which is what makes the unseen
$X$ estimable from quotes. Second, nothing converges pathwise: the limit in
\eqref{eq:X} is taken over the forecast horizon $h$, not over calendar
time; what is stable, stationary rather than convergent, is the distance
between the two prices. The phenomenon of this paper
is that pull; the question is what trading it is worth. As
\eqref{eq:wealth} below makes precise, the strategy's profit will not
depend on where $X$ goes, only on when the mid strays from it.

\subsection{The book, from the bid, and the parity lock}

With tick size $\delta$ (the smallest price increment), the displayed bid
$B_t$ lives on the tick grid, and for a large-tick asset the spread $S_t$
takes two values:
\begin{equation}\label{eq:grid}
B_t\in\delta\mathbb Z,\qquad S_t\in\{\delta,2\delta\},\qquad
M_t=B_t+\tfrac{S_t}{2}\in\tfrac{\delta}{2}\mathbb Z .
\end{equation}
This is the regime of liquid, large-tick assets, and it is worth saying
who lives in it. When the queue at the touch is deep, the spread sits at
one tick almost all of the time; a market order that consumes the best
queue, or a cancellation that empties it, opens the spread to two ticks
for a short while, until a new limit order closes it: the opens and
closes of Figure~\ref{fig:chain}, brief by nature
\cite{dayrirosenbaum15}. Wider spreads would require emptying two or
more price levels before any replenishment arrives; in this class the
levels behind the touch are deep and the inside is refilled fast. This
is not a claim that any single refilling order profits (adverse
selection decides that case by case, and the spread opens precisely
when the book has just been hit) but the defining regularity of the
class: the tick is large relative to the spread competition would set,
so the one-tick spread is worth queueing for, priority at the inside is
contested, and spreads of three ticks or more are empirically
negligible \cite{dayrirosenbaum15}. We exclude them from the state
space. Assets that are illiquid, or whose
tick is small relative to their volatility, carry many spread values at
once and fall outside the model. Modelling the bid rather than the mid is not
cosmetic: the mid is a derived
average, and once the spread can open and close, the mid must move in
half-ticks. That grid arithmetic has a structural consequence.

\begin{fact}[The parity lock]\label{fact:parity}
Measure the mid in half-ticks, so that it is an integer; \emph{parity}
means whether that integer is even or odd; equivalently, in ticks, whether
the mid sits at a half-integer or at an integer. A quote slide moves the mid by
two half-ticks, preserving parity; opening or closing the spread moves it
by one, flipping parity; and an open is exactly the event that sets
$S=2\delta$. One direction is pure grid arithmetic:
\begin{equation}\label{eq:parity}
S=\delta\ \text{(tight)}\implies M\in\delta\mathbb Z+\tfrac{\delta}{2},
\qquad
S=2\delta\ \text{(open)}\implies M\in\delta\mathbb Z .
\end{equation}
Within the two-valued class \eqref{eq:grid} the implication inverts (two
spread values, two parities): the spread is the parity of the mid, not a
second state variable. What survives is one continuous coordinate, the
gap, plus a parity bit that flips at half-tick moves
(Figure~\ref{fig:chain}).
\end{fact}

\begin{figure}[!tb]
\centering
\begin{tikzpicture}[>=Stealth, node distance=20mm,
  lvl/.style={circle,draw,minimum size=12.5mm,inner sep=0pt,font=\scriptsize},
  lbl/.style={font=\scriptsize, text=black!60}]
\node[lvl] (m2) {$M_t{-}\delta$};
\node[lvl, right=of m2] (m1) {$M_t{-}\tfrac{\delta}{2}$};
\node[lvl, right=of m1] (z)  {$M_t$};
\node[lvl, right=of z]  (p1) {$M_t{+}\tfrac{\delta}{2}$};
\node[lvl, right=of p1] (p2) {$M_t{+}\delta$};
\node[lbl, below=1.5mm of m2] {tight};
\node[lbl, below=1.5mm of m1] {open};
\node[lbl, below=1.5mm of z]  {tight};
\node[lbl, below=1.5mm of p1] {open};
\node[lbl, below=1.5mm of p2] {tight};
\draw[->] (m2) to[bend left=18] node[above,font=\scriptsize]{open} (m1);
\draw[->] (m1) to[bend left=18] node[below,font=\scriptsize]{close} (m2);
\draw[->] (m1) to[bend left=18] node[above,font=\scriptsize]{close} (z);
\draw[->] (z)  to[bend left=18] node[below,font=\scriptsize]{open} (m1);
\draw[->] (z)  to[bend left=18] node[above,font=\scriptsize]{open} (p1);
\draw[->] (p1) to[bend left=18] node[below,font=\scriptsize]{close} (z);
\draw[->] (p1) to[bend left=18] node[above,font=\scriptsize]{close} (p2);
\draw[->] (p2) to[bend left=18] node[below,font=\scriptsize]{open} (p1);
\draw[->] (m2) to[bend left=42] node[above,font=\scriptsize]{slide} (z);
\draw[->] (z)  to[bend left=42] node[above,font=\scriptsize]{slide} (p2);
\draw[->] (z)  to[bend left=55] node[below,font=\scriptsize]{slide} (m2);
\draw[->] (p2) to[bend left=55] node[below,font=\scriptsize]{slide} (z);
\end{tikzpicture}
\caption{The parity lock: the values the mid can reach on the half-tick
grid around a tight mid $M_t$. Tight (half-integer) levels
slide a full tick, preserving parity, or open a half-tick, flipping it;
open (integer) levels close only. The spread is the parity of the mid: it is not a second state variable,
and the gap is the only continuous coordinate. The grid continues in both
directions; only the moves reachable in one step from $M_t$ are drawn.}
\label{fig:chain}
\end{figure}
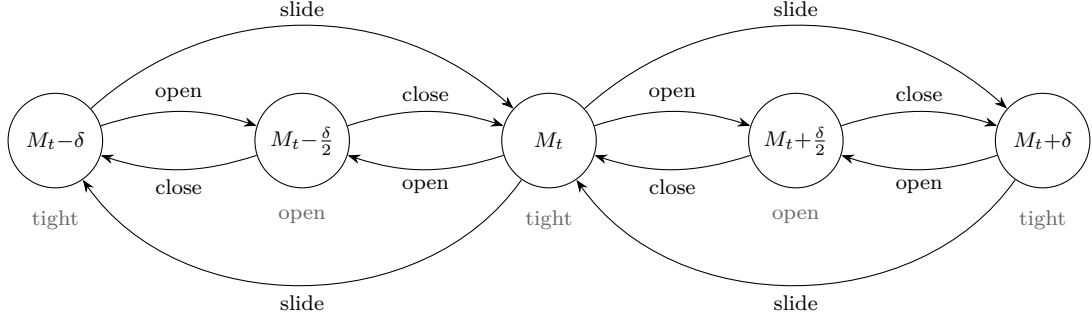

\subsection{The dynamics}
\label{sec:jumpmodel}

\begin{definition}[Efficient price]\label{def:X}
The efficient price is the exogenous Brownian martingale
\begin{equation}\label{eq:dX}
dX_t=\sigma_X\,dZ_t ,
\end{equation}
with $Z$ a standard Brownian motion and $\sigma_X>0$. Nothing in the
book enters its dynamics.
\end{definition}

\begin{definition}[The book]\label{def:M}
The book is driven by six counting processes, the moves of
Figure~\ref{fig:chain}: slides $N^{s\uparrow},N^{s\downarrow}$, opens
$N^{o\uparrow},N^{o\downarrow}$, closes
$N^{c\uparrow},N^{c\downarrow}$, each with its own
$(\mathcal F_t)$-intensity. With $G_{t-}=M_{t-}-X_{t-}$, $x^{\pm}=\max(\pm x,0)$, and the spread
read off the parity of the mid (Fact~\ref{fact:parity}),
\begin{equation}\label{eq:intensity}
\begin{aligned}
\lambda^{s\uparrow}(t)&=\mathbf 1\{S_{t-}=\delta\}
\Big[\mu_s+\tfrac{2\alpha_s}{\delta}\,G_{t-}^{-}\Big],
&\lambda^{s\downarrow}(t)&=\mathbf 1\{S_{t-}=\delta\}
\Big[\mu_s+\tfrac{2\alpha_s}{\delta}\,G_{t-}^{+}\Big],\\[2pt]
\lambda^{o\uparrow}(t)&=\mathbf 1\{S_{t-}=\delta\}
\Big[\mu_o+\tfrac{2\alpha_o}{\delta}\,G_{t-}^{-}\Big],
&\lambda^{o\downarrow}(t)&=\mathbf 1\{S_{t-}=\delta\}
\Big[\mu_o+\tfrac{2\alpha_o}{\delta}\,G_{t-}^{+}\Big],\\[2pt]
\lambda^{c\uparrow}(t)&=\mathbf 1\{S_{t-}=2\delta\}
\Big[\mu_c+\tfrac{2\alpha_c}{\delta}\,G_{t-}^{-}\Big],
&\lambda^{c\downarrow}(t)&=\mathbf 1\{S_{t-}=2\delta\}
\Big[\mu_c+\tfrac{2\alpha_c}{\delta}\,G_{t-}^{+}\Big].
\end{aligned}
\end{equation}
That $\lambda^{m}$ is the $(\mathcal F_t)$-intensity of $N^{m}$ means
\begin{equation}\label{eq:poisson}
N^{m}_t-\int_0^{t}\lambda^{m}(u)\,du \ \text{ is an }(\mathcal F_t)\text{-local
martingale},\qquad m\in\{s\!\uparrow,\dots,c\!\downarrow\},
\end{equation}
and no two of the six processes jump at the same time. The mid is then
the pure jump process
\begin{equation}\label{eq:mid}
M_t=M_0+\delta\big(N^{s\uparrow}_t-N^{s\downarrow}_t\big)
+\tfrac{\delta}{2}\big(N^{o\uparrow}_t-N^{o\downarrow}_t\big)
+\tfrac{\delta}{2}\big(N^{c\uparrow}_t-N^{c\downarrow}_t\big).
\end{equation}
Throughout, the baselines are positive and the correction is live:
$\mu_s,\mu_o,\mu_c>0$ and $\alpha_s,\alpha_o\ge0$, $\alpha_c>0$.
Appendix~\ref{app:prelim}, (P0), records that a nonexplosive process
with these characteristics exists.
\end{definition}

\paragraph{Why this specification.} None of \eqref{eq:dX}--%
\eqref{eq:mid} is invented for the occasion. That the event
intensities of the best quotes depend on the state of the book is the
tradition of Markovian and queue-reactive order-book models: Cont and
de~Larrard \cite{contdelarrard13} move the quotes at state-dependent
rates, and the queue-reactive model of Huang, Lehalle and Rosenbaum
\cite{huanglehalle15} re-centres the book around a reference price;
both are fitted to exchange tick data, not merely posited. That
the moves of a grid-valued price are driven by its distance to a latent
efficient price is the model with uncertainty zones of Robert and
Rosenbaum \cite{robertrosenbaum11}, where a move fires when $X$ crosses a
threshold; and order flow whose intensity is driven by a latent Brownian
efficient price is the Cox-process construction of Delattre, Robert and
Rosenbaum \cite{delattre13}, built there for estimation. Our linear
intensities are the smooth analogue of the first and the control-ready
version of the second: the minimal specification delivering an exact
linear compensator. The same architecture, a book whose event
intensities are driven by a latent Brownian efficient price, has since
been developed into a unified queue-reactive framework by Sfendourakis
\cite{sfendourakis25}, with stability of the mid around the efficient
price and diffusive limits at the large scale; and Pulido, Rosenbaum and
Sfendourakis \cite{pulido26} solve a full-information market-making
problem on the uncertainty-zone member of the family, the maker who
knows the efficient price steering her quoted volumes. What is new here
is not the architecture but what it is made to yield: the parity lock
collapses the state to one coordinate, and the position itself, not
the quotes, is the control, solved to a formula. The crossing of the
mutually exciting models \cite{rabechini19,bacry13a} is carried here by
the gap: a jump up lifts $M$, hence shifts the balance of intensities
against the move: it raises the downward ramp when the gap is positive,
releases the upward ramp when the gap is negative, and does both when
the gap crosses zero. Just as important, a fundamental
that drifts away moves the intensities with it, even while the book
stands still. (A transient Hawkes channel with kernel
$\alpha e^{-\beta(t-s)}$, as in \cite{rabechini19}, can be layered on
top; it adds short-horizon retracement without changing what follows.)
Every intensity is nonnegative by construction. One
\emph{balanced-response} condition equalises the book's corrective
drift across parities,
\begin{equation}\label{eq:balanced}
2\alpha_s+\alpha_o=\alpha_c=:\alpha :
\end{equation}
the book responds to the gap, in displacement, however the correction is
delivered. This is a genuine restriction, not a normalisation: without it
the drift is parity-dependent, $-(2\alpha_s+\alpha_o)G$ in a tight book
and $-\alpha_c G$ in an open one, and the single-rate reversion of
Proposition~\ref{prop:reversion} below would split in two. The condition's bite, however, is
bounded by a quantity the tape reports.

\paragraph{Why one-sided ramps, and not an exponential kernel.} A
natural objection to \eqref{eq:intensity} is its asymmetry: when the
book dislocates, the corrective side intensifies while the opposing side
merely sits at its baseline, rather than being suppressed. The choice is
forced, and it is worth saying by what. (a) The drift of the gap sees
only the \emph{difference} of the two sides, and for the one-sided
ramps that difference is exactly linear everywhere:
$G^{+}-G^{-}=G$ identically, so the conditional mean of
Proposition~\ref{prop:reversion} is exact for all $G$, not merely near
zero. Any pair of kernels with the same difference produces the same
drift; what distinguishes them is the \emph{sum}, the total event
rate, which is a variance channel, not a drift channel. (b)
Suppressing the opposing side linearly drives its intensity negative at
finite dislocation; truncating at zero restores nonnegativity but kinks
the difference, destroying the exact linearity, and forbids wrong-way
moves in a dislocated book, which do occur and which the baseline
$\mu$ of the passive side is there to produce. (c) The level is not
a free convenience: one-sided ramps make the total event rate increase
with $|G|$, the two-sided version holds it constant, and the two are
distinguishable on a tape: activity clusters when the book is
dislocated, which is the queue-reactive regularity
\cite{huanglehalle15} carried to the efficient-price setting in
\cite{sfendourakis25}, and selects the one-sided form. The growing
rate also makes silence informative: a quiet book is evidence of a
small gap, the channel the filtering continuation of
Section~\ref{sec:remarks} uses. (d) An
exponential kernel $\mu e^{\beta G^{\mp}}$ agrees with the ramp to
first order over the range the process actually samples ($|G|$ of a
few $s_G$, itself of the order of the tick), so within that range it
adds one curvature parameter the data cannot identify, while destroying
the exact conditional mean, the moment identity of
Proposition~\ref{prop:reversion}, and the closed form of
Section~\ref{sec:band}. The linear ramp is the first-order theory of every smooth monotone
kernel, and the only member of the family for which the reversion is a
theorem rather than an expansion: a one-sided kernel has an exactly
linear difference only if it is itself linear.

\paragraph{What the condition costs.} Write
$\alpha_{\rm tight}=2\alpha_s+\alpha_o$ for the tight-book rate and let
\begin{equation}\label{eq:popen}
p\;=\;\mathbb P(S=2\delta),\qquad
\frac{p}{1-p}
=\frac{\mu_o+\tfrac{\alpha_o}{\delta}\,\E_\pi\big[\,|G|\;\big|\;S=\delta\,\big]}
      {\mu_c+\tfrac{\alpha_c}{\delta}\,\E_\pi\big[\,|G|\;\big|\;S=2\delta\,\big]}
\end{equation}
be the stationary fraction of time the book is open. The ratio is an identity,
not an expansion: in stationarity the flow of opens out of tight books balances
the flow of closes out of open ones, and each flow is its baseline plus its ramp
averaged over the parity it acts on. Dropping the ramps leaves
$\mu_o/(\mu_o+\mu_c)$, the reading at $G\equiv0$; it is accurate only when
$\alpha_o\E_\pi|G|\ll\delta\mu_o$ and $\alpha_c\E_\pi|G|\ll\delta\mu_c$. The
first binds, because it is $\mu_o$ that the class makes small: once
$\E_\pi|G|$ is of the order of a tick, a ramp with $\alpha_o>0$ is of the size
of its own baseline, and the naive ratio is then out by a factor of several,
in either direction, since the ramp on opens pushes $p$ up and the ramp on
closes pushes it down. The parity
relaxes at rate $2(\mu_o+\mu_c)$ or faster, and for a liquid large-tick asset
closes are fast:
priority at the inside is contested, so an open invites the
next refill. Whether that refill profits is an adverse-selection
question the model does not adjudicate. Thus $\mu_c\gg\mu_o$ stands as the
empirical regularity of the class
\cite{dayrirosenbaum15,contdelarrard13}, and $p$ is small: what has to be
small against $\mu_c$ is the whole numerator of \eqref{eq:popen}, baseline and
ramp together. Averaging the parity-dependent drift over
its occupation gives an effective rate
\begin{equation}\label{eq:aeff}
\alpha_{\rm eff}
=(1-p)\,\alpha_{\rm tight}+p\,\alpha_c
=\alpha_{\rm tight}+p\,(\alpha_c-\alpha_{\rm tight}),
\end{equation}
so relaxing \eqref{eq:balanced} moves the reversion rate by at most
$p\,|\alpha_c-\alpha_{\rm tight}|$. Two readings follow. First,
\eqref{eq:balanced} buys \emph{exactness}: with it,
Proposition~\ref{prop:reversion} holds identically rather than after
averaging, and that is what it is for. Second, it does not carry the
result. For a liquid large-tick asset the book is tight almost always
($p$ of order one per cent or less), so the split the condition prevents is
a correction of that order, and under \eqref{eq:balanced} the fitted
$\alpha$ is properly read as the parity-weighted rate \eqref{eq:aeff},
which is the object the trading problem of Section~\ref{sec:band} needs.
Equation \eqref{eq:aeff} is the standard averaging statement for a linear
drift modulated by a finite-state chain that is fast against the drift it
modulates, here $2(\mu_o+\mu_c)\gg\alpha$, which the same class assumption
delivers; we assert it as a bound on
what \eqref{eq:balanced} excludes, and impose \eqref{eq:balanced} in what
follows.

The virtue of \eqref{eq:popen} is that $p$ is not a free number to be
asserted. It is the fraction of time the displayed spread is two ticks,
which a book feed reports directly, and it is simultaneously an output of
the specification through $\mu_o/\mu_c$, so the model can be checked
against the tape twice over.
Figure~\ref{fig:path} shows the resulting dynamics in motion.

\begin{figure}[!tb]
\centering
\includegraphics[width=.9\textwidth]{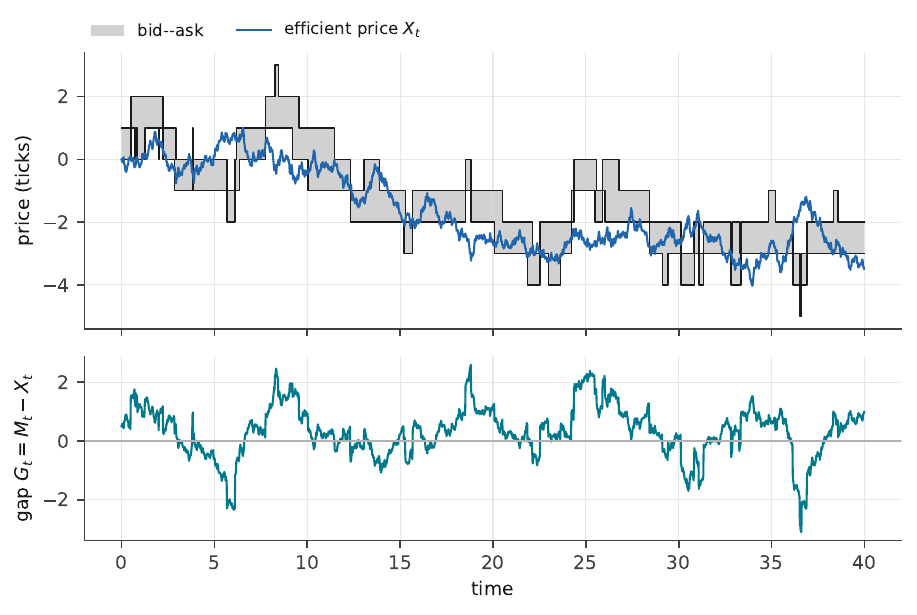}
\caption{The model in motion: a sample path of
Definitions~\ref{def:X}--\ref{def:M} with illustrative parameters;
nothing is calibrated. Top: the bid and ask (the grey ribbon) move on
the tick grid, one tick wide most of the time and two for short episodes,
the spread opening and closing on its own, while the efficient price
threads through it. The baselines here open the book far more often than the
class of Section~\ref{sec:model} does, so that the episodes are visible at all:
this path is open about eighteen per cent of the time, against the one per cent
or less argued above for the class. Bottom: the gap $G_t=M_t-X_t$:
stationary, breathing around zero, the object the rest of the paper
prices.}
\label{fig:path}
\end{figure}

\subsection{Derived properties: reversion, ergodicity, moments}
\label{sec:derived}

The central algebra of the paper is one cancellation. Within each pair
of intensities in Definition~\ref{def:M} the baselines cancel and,
since $G^{-}-G^{+}=-G$ identically, the kinks cancel too: each pair's
net drift is exactly linear in the gap. Weighing each move by its size,
\begin{equation}\label{eq:compensator}
\E[dM_t\mid\mathcal F_{t-}]=
\begin{cases}
\Big(\delta\cdot\tfrac{2\alpha_s}{\delta}
+\tfrac{\delta}{2}\cdot\tfrac{2\alpha_o}{\delta}\Big)
\big(G^{-}_{t-}-G^{+}_{t-}\big)\,dt
=-(2\alpha_s+\alpha_o)\,G_{t-}\,dt, & S_{t-}=\delta,\\[5pt]
\tfrac{\delta}{2}\cdot\tfrac{2\alpha_c}{\delta}\,
\big(G^{-}_{t-}-G^{+}_{t-}\big)\,dt
=-\alpha_c\,G_{t-}\,dt, & S_{t-}=2\delta,
\end{cases}
\end{equation}
and the balanced-response condition \eqref{eq:balanced} makes the two
lines equal: $\E[dM_t\mid\mathcal F_{t-}]=-\alpha\,G_{t-}\,dt$ in either
parity, exactly. The positive parts
exist to keep each intensity nonnegative; they leave no trace in the
drift, and the nonlinearity survives only in the total event rate, where
the stationary mean defining $\sigma_M^{2}$ picks it up. Since $X$ is a
martingale, $\E[dX_t\mid\mathcal F_{t-}]=0$, and the reversion of the gap follows.
We record it together with the stationary structure it forces.

\begin{proposition}[Exact reversion, ergodicity, and moments]\label{prop:reversion}
Under Definitions~\ref{def:X}--\ref{def:M} and the balanced-response condition
\eqref{eq:balanced}, the gap mean-reverts exactly,
\begin{equation}\label{eq:drift}
\boxed{\;\E\big[dG_t\mid\mathcal F_{t-}\big]=-\alpha\,G_{t-}\,dt\;}
\end{equation}
with no expansion and no approximation, and the conditional mean closes at every
horizon, from every state, with no stationarity involved:
\begin{equation}\label{eq:condmean}
\E[G_{t+h}\mid\mathcal F_t]=G_t\,e^{-\alpha h},\qquad t,h\ge0 .
\end{equation}
The pair $(G,S)$ is Markov (the intensities and the diffusion are functions of
the pair alone, and each move updates the spread deterministically by the parity
lock) and it is positive Harris recurrent with a unique invariant law $\pi$ of
finite second moment; under $\pi$,
\begin{equation}\label{eq:moments}
\E_\pi[G]=0,\qquad
\sG^{2}:=\mathrm{Var}_\pi(G)=\frac{\sigma_X^{2}+\sigma_M^{2}}{2\alpha},\qquad
\mathrm{Cov}_\pi\big(G_t,G_{t+h}\big)=\sG^{2}\,e^{-\alpha h},
\end{equation}
where $\sigma_M^{2}$ is the stationary mean of the book's jump-variance rate
$\sum_m\lambda^{m}(t)\Delta_m^{2}$, with $\Delta_m$ the mid displacement of
move $m$: $\pm\delta$ for slides, $\pm\delta/2$ for opens and closes. No
dynamics for $G$ is assumed: every identity here
is a consequence of \eqref{eq:intensity}--\eqref{eq:balanced}.
\end{proposition}

\begin{proof}
The reversion is the one-line cancellation \eqref{eq:compensator} closed by
\eqref{eq:balanced}; ergodicity is a Foster--Lyapunov argument with $V=G^{2}$, the
Brownian component supplying the required irreducibility; and
the moments follow from It\^o's formula on $G^{2}$ together with the linearity of
\eqref{eq:drift}, which closes the conditional mean at every horizon and hence the
autocovariance. Every step is written out in Appendices~\ref{pf:drift}
and~\ref{pf:moments}.
\end{proof}

\begin{remark}[$\sG$ in closed form, to within a bracket]\label{rem:bracket}
The middle identity of \eqref{eq:moments} fixes $\sG$ only implicitly:
$\sigma_M^{2}$ is a stationary mean, and through the ramps it depends on
$\E_\pi|G|$. Two elementary bounds close it. Write
\[
\sigma_{M,0}^{2}:=\delta^{2}\Big[(1-p)\big(2\mu_s+\tfrac{\mu_o}{2}\big)
+p\,\tfrac{\mu_c}{2}\Big],
\qquad
b:=\delta\Big(\alpha-\tfrac{\alpha_o}{2}\Big),
\]
the jump-variance rate the baselines alone produce, and the largest weight a
ramp carries in either parity. Since each ramp contributes a nonnegative
multiple of $|G|$, and no more than $b\,|G|$,
$\sigma_{M,0}^{2}\le\sigma_M^{2}\le\sigma_{M,0}^{2}+b\,\E_\pi|G|$; and
$\E_\pi|G|\le\sG$ by Cauchy--Schwarz. Substituting in \eqref{eq:moments} and
solving the two quadratics in $\sG$,
\[
\sqrt{\frac{\sigma_X^{2}+\sigma_{M,0}^{2}}{2\alpha}}
\;\le\;\sG\;\le\;
\frac{b+\sqrt{b^{2}+8\alpha\big(\sigma_X^{2}+\sigma_{M,0}^{2}\big)}}{4\alpha}\,,
\]
both ends explicit in the primitives once $p$ is read off the tape. The two
ends meet as the ramps' share of the event rate vanishes, the ratio that
also governs \eqref{eq:popen}. For the trading problem none of this is needed:
Section~\ref{sec:band} uses the pair $(\alpha,\sG)$, which the gap's
autocovariance \eqref{eq:moments} identifies directly.
\end{remark}

The forecast property
\eqref{eq:X} promised in Section~\ref{sec:evidence} is now a theorem: by
\eqref{eq:condmean}, $\E[G_{t+h}\mid\mathcal F_t]=G_t\,e^{-\alpha h}\to0$, so
$\lim_{h}\E[M_{t+h}\mid\mathcal F_t]=X_t$: the fundamental is precisely the
permanent component of the observed price. The anatomy respects the causality
throughout: between events $dG=-\sigma_X\,dZ$, driftless, and the pull is realised
entirely at the jumps of $M$, whose
intensities lean toward the fundamental. At an event the gap is kicked by the full
move: $\pm\delta$ for a slide, $\pm\delta/2$ for an open or a close, so a
tight book kicks by either size, an open book only by the smaller one. We treat $G$ as observable,
the full-information case; Section~\ref{sec:remarks} returns to this assumption.

\subsection{The trading problem}

A marketable order pays, relative to the mid, the half-spread
\begin{equation}\label{eq:phi}
\phi(S)=\frac{S}{2}\,,
\end{equation}
which is the only cost in the model, the cost the book itself
imposes. Write
$\phi_t:=\phi(S_t)=S_t/2$ for its value at time $t$: a two-valued
process, $\delta/2$ in a tight book and $\delta$ in an open one. (Any
per-lot fee simply adds to it.) One assumption closes the loop with Definition~\ref{def:M}: the trader
is \emph{small}. Orders fill at the displayed quotes and do not enter the
intensities \eqref{eq:intensity}. Our earlier work models precisely the impact of orders
large enough to violate this \cite{rabechini19}; the provider's response
to being traded against is deferred to Section~\ref{sec:remarks}.
The strategy is the inventory process $q=(q_t)_{t\ge0}$. We call it
\emph{admissible} if
it is adapted with right-continuous paths, takes values in $[-1,+1]$, and
is piecewise constant with a finite expected number of trades on bounded
intervals. The cap is a risk limit, and fractional positions are allowed. An
independent randomisation drawn at time zero is also allowed; it alters nothing
about the book. The trades
react to the state the instant it moves (the execution convention below),
so $q$ itself is not predictable; the integrands in \eqref{eq:wealth}
are the left limits $q_{t-}$, which are. The trader starts flat, and trades are the jumps of $q$: buys and sells
never coincide, so $dq_t>0$ is a buy, $dq_t<0$ a sell, and $|dq_t|$ lots
change hands. Execution is at the touch, a buy paying the ask
$M_t+\phi_t$ and a sell receiving the bid $M_t-\phi_t$, and we track two
markings of the same book: wealth marked to the mid, the desk's
convention, and marked to the efficient price, the analyst's:
\begin{equation}\label{eq:wealth}
\begin{aligned}
dW_t &= q_{t-}\,dM_t-\phi_t\,|dq_t|,\\[2pt]
dW^{X}_t &= q_{t-}\,dX_t-G_t\,dq_t-\phi_t\,|dq_t| .
\end{aligned}
\end{equation}
The first reads as a desk reads it: the position rides the mid, and
every trade pays the half-spread. The second is the same cash flows with
the inventory valued at $X$: per lot, a buy pays the gap plus the
half-spread and a sell collects the gap minus it, and the trade terms
are functions of the state alone. The trade term of $dW^{X}$ carries the \emph{post-jump} gap $G_t$: the trader
sees the move and then lifts the quote it left behind. With that convention the
two markings differ by $W_t-W^{X}_t=q_tG_t$ exactly
(Appendix~\ref{pf:wealth}), a difference bounded in $L^{1}$ uniformly in
$t$, so the long-run objective below is identical on
either; and on $W^{X}$ the inventory term is a zero-mean martingale:
$q_{t-}$ is predictable and bounded, $X$ is a square-integrable martingale, so
$\E\!\int q_{t-}\,dX=0$. Every formula in this paper is a statement about
the mid-marked book, and \emph{the problem is the timing of the gap, not
the forecasting of $X$}.

The objective is the long-run average profit rate, evaluated conservatively by
the lower limit (for an arbitrary strategy the average need not converge; for
the band it does, and the $\liminf$ is attained: Proposition~\ref{prop:rate} and
Appendix~\ref{pf:rate}, Step 6)
\begin{equation}\label{eq:problem}
\sup_{q\ \mathrm{admissible}}\ \liminf_{T\to\infty}\frac{\E[W_T]}{T}.
\end{equation}

One reduction can be dispatched at once: the search over all inventories
in $[-1,+1]$, a far larger and harder optimisation, collapses to the
three values $\{-1,0,+1\}$, with nothing lost. Unlike the reductions of
the next section, this is a pathwise statement about the exact jump
process, with no surrogate anywhere in its proof.

\begin{proposition}[Full positions suffice]\label{prop:bangbang}
For an admissible strategy $q=(q_t)_{t\ge0}$ with values in $[-1,+1]$,
define its \emph{layers}: for each level $u\in(0,1]$,
\[
a^{u}_t \;:=\; \mathbf 1\{q_t\ge u\}\;-\;\mathbf 1\{q_t\le -u\} .
\]
Each layer is an admissible strategy, and, pathwise and at every horizon
$T$,
\begin{equation}\label{eq:layer}
W_T(q)\;=\;\int_0^1 W_T(a^{u})\,du ,
\end{equation}
and consequently
\begin{equation}\label{eq:supeq}
\sup_{\substack{q\ \mathrm{admissible}\\ q_t\in[-1,+1]}}\
\liminf_{T\to\infty}\frac{\E[W_T(q)]}{T}
\;=\;
\sup_{\substack{q\ \mathrm{admissible}\\ q_t\in\{-1,0,+1\}}}\
\liminf_{T\to\infty}\frac{\E[W_T(q)]}{T}\ .
\end{equation}
\end{proposition}

The proof (Appendix~\ref{pf:bangbang}) is a layer decomposition. The
layers at levels $u\le|q_t|$ hold the full position, with the sign of
$q_t$; the others are flat ($a^u_t=0$); and together they average back to the
position,
$q_t=\int_0^1 a^u_t\,du$. Gains are linear in the position,
so they average exactly across layers; by a coarea identity, so does the
proportional cost; and drawing the level at random at time zero gives a
single three-valued strategy with the same expected wealth at every
horizon, which forces \eqref{eq:supeq}. The reduction rests on linear
reward and proportional cost: under risk aversion or superlinear
trading costs, interior positions can be strictly optimal. From here on
$q$ takes values in $\{-1,0,+1\}$: a conclusion, not a constraint.

\section{The optimal band}\label{sec:band}

Three further reductions precede the choice of class, and they do not
stand on the same footing. \emph{The price level is irrelevant information}: the
pair $(G,S)$ is autonomous (its intensities and diffusion involve the
gap and the parity, never the level of $M$ or $X$), and by
\eqref{eq:wealth} the conditional law of every future wealth increment
given $\mathcal F_t$ depends on the past only through
$(G_t,S_t,q_{t-})$; conditioning a strategy on the level, or on any
other adapted quantity, adds no information about future rewards, so the
control problem lives on the state $(G,S,q)$. \emph{Time}: the model is
time-homogeneous and \eqref{eq:problem} an ergodic average, so the
optimum is sought among stationary state policies, the standard
reduction for ergodic criteria and part of what the verification theory
below delivers on the surrogate. \emph{Symmetry}: Definition~\ref{def:M}
is invariant under $(G,\uparrow)\mapsto(-G,\downarrow)$ and
\eqref{eq:wealth} under exchanging buys with sells, so every policy has
a mirror image with the same rate, and a symmetric optimum exists. We
therefore look for the optimum of \eqref{eq:problem} among symmetric
threshold strategies: for a band half-width $\theta>0$,
\begin{equation}\label{eq:band}
q^{\theta}_t=
\begin{cases}
+1, & G_t\le-\theta \quad\text{(buy: the mid is cheap),}\\
-1, & G_t\ge+\theta \quad\text{(sell: the mid is dear),}\\
q^{\theta}_{t-}, & \text{otherwise (hold, no churn).}
\end{cases}
\end{equation}
Read \eqref{eq:band} as a \emph{target position}, not an order: when
$G_t$ touches $-\theta$ or $+\theta$ the order sent is the difference
$q^{\theta}_t-q^{\theta}_{t-}$, one lot at the first entry and two lots at every fill
thereafter, flipping the position through zero without ever resting
there.

\paragraph{Why thresholds.} The class is not an ansatz. Match each sold
lot to the earliest not-yet-matched bought lot, and each buy-to-cover to
the earliest unmatched short sale: this pairs the executed lots into
\emph{round trips} (one lot bought and one lot sold: a position
opened and later closed, in either order), leaving at most one lot
unpaired, the position still open at $T$. Let $n_T$ be the
number of round trips completed by
time $T$; for the $i$-th, let $\tau^{b}_i$ and $\tau^{s}_i$ be the fill
times of its buy and sell legs, and write $G^{b}_i=G_{\tau^{b}_i}$,
$G^{s}_i=G_{\tau^{s}_i}$ for the gap and
$\phi^{b}_i=\phi_{\tau^{b}_i}$, $\phi^{s}_i=\phi_{\tau^{s}_i}$ for the
half-spread at each leg. Summing the trade bookings of
\eqref{eq:wealth}, the wealth of \emph{every} three-valued strategy,
the class to which Proposition~\ref{prop:bangbang} reduces the problem,
decomposes as
\begin{equation}\label{eq:pairs}
W^{X}_T=\sum_{i=1}^{n_T}
\big(G^{s}_i-G^{b}_i-\phi^{b}_i-\phi^{s}_i\big)
\;+\;\int_0^{T}\! q_{s-}\,dX_s\;+\;O(1),
\end{equation}
term by term, in the order displayed: one summand per completed round
trip, worth \emph{the traverse of the gap minus both half-spreads}
whether the trip is long or short; a zero-mean
martingale, the bounded position integrated against the martingale
$X$; and a remainder, the still-open trip marked to market: at most
one lot times a gap of bounded mean (P2), an $O(1)$ that vanishes
against $T$.

Take expectations in \eqref{eq:pairs} and divide by $T$: the martingale
has mean zero, the remainder vanishes against $T$, and only the trips
survive,
\[
\frac{\E[W_T]}{T}
=\frac{1}{T}\,\E\Big[\sum_{i=1}^{n_T}
\big(G^{s}_i-G^{b}_i-\phi^{b}_i-\phi^{s}_i\big)\Big]+o(1) .
\]
The long-run rate is trips per unit time, times what a trip pays. A
strategy therefore earns only through \emph{where it places its two
legs}, and we look for the optimum among threshold rules: trade when the gap
clears a threshold, one per side; the depth itself is the object
of the optimisation, since a deeper band earns more per trip but waits
longer between trips.

How far can the restriction to the band be trusted? The answer differs
between the paper's two levels, and it should be kept straight. Proposition~\ref{prop:bangbang} is exact and pathwise; its layer
argument reads verbatim on the surrogate, so ``all admissible
strategies'' below includes fractional inventories. The step from all
admissible strategies down to the flip band is different: it is a
theorem only on the Gaussian surrogate introduced below \eqref{eq:ou},
assembled from the literature in three pieces. \emph{Thresholds}: for a
mean-reverting diffusion under proportional costs, optimal entry and
exit are first-passage thresholds, over all admissible strategies,
time-varying and level-dependent ones included; the verification proofs
are \cite{zhangzhang08} for a positive mean-reverting asset,
\cite{zervos13} for general one-dimensional diffusions, and
\cite{leungli16} for the OU asset with a complete analysis. \emph{The criterion}: those theorems discount; the
long-run average of \eqref{eq:problem} is the renewal--reward evaluation
of the trading cycle, as in \cite{bertram10}, with ergodic verification
along the vanishing-discount lines of \cite{abg12}. \emph{Symmetry and no flat
zone}: mirror invariance forces the symmetric band, and
Appendix~\ref{pf:flat} rules out resting flat between the legs. On the exact jump process the same reduction is
a conjecture supported by the same structure (Markov state, ergodic
criterion, monotone trip payoff, mirror symmetry): a verification theorem for the parity-locked jump process is open,
and the simulations of Section~\ref{sec:validate} test the surrogate's
timing within the band class, not the optimality of the class.

The
thresholds can also be taken parity-blind: the refinement
$\theta(S)$ moves the rate by $O(p)$, with $p$ the open-book occupancy
\eqref{eq:popen}, the same order as the cost freeze below.
Appendix~\ref{pf:flat} also records the exact accounting: an excursion
to flat is worth, in expectation, $q(G_a-G_b)-2\phi$ relative to
holding, and an instantaneous round trip pays the full spread for a
traverse of zero. (A risk-averse criterion would reopen a flat zone
around zero; we do not pursue it.)

Hence \eqref{eq:band}, with a single
number left to choose: $\theta$. Proposition~\ref{prop:rate} prices the
band exactly, carrying the half-spread $\phi_F$ realised at each fill;
one convention then lightens the optimisation. Fills in open books are
the exception, the open state's occupancy being the $p$ of \eqref{eq:popen}; from
\eqref{eq:ratesur} on we freeze
the cost at its tight-book value, $\phi:=\delta/2$, returning to what
the freeze drops in Section~\ref{sec:remarks}.

\begin{proposition}[The rate of a band, exactly]\label{prop:rate}
In the model of Definitions~\ref{def:X}--\ref{def:M}, the band
\eqref{eq:band} earns the long-run rate
\begin{equation}\label{eq:rate}
R(\theta)=\frac{2\big(\E_{\pi_F}[\,|G_F|\,]-\E_{\pi_F}[\phi_F]\big)}{m(\theta)}\,,
\end{equation}
where $\pi_F$ is the stationary law of the post-fill chain, $m(\theta)$ the
stationary expected time between opposite fills, and $G_F$ and $\phi_F$ the gap
and the half-spread at a fill, with $\theta\le\E_{\pi_F}|G_F|<\theta+\delta$: the
gap leaves $(-\theta,\theta)$ either by diffusion, exactly at the edge, or
by a jump, past it.
\end{proposition}

\begin{proof}
The states of $(G,S,q)$ observed at successive fills form a Markov chain; a
diffusive-crossing construction gives it an accessible atom and a uniformly
finite mean inter-fill time, hence positive Harris recurrence with invariant law
$\pi_F$, and $m(\theta)$ is the mean inter-fill time under $\pi_F$. Each fill
flips the inventory by two lots, each lot earning $|G_F|-\phi_F$, so a cycle
carries reward $4\,\E_{\pi_F}[|G_F|-\phi_F]$ and has mean length $2m(\theta)$;
the regenerative ratio gives \eqref{eq:rate}. The two markings agree because
$W-W^{X}=qG$ is bounded in $L^{1}$, and on $W^{X}$ the inventory term
$\int q_{s-}\,dX_s$ is a zero-mean martingale. The bound on $\E_{\pi_F}|G_F|$ is
the maximal jump, $\delta$. Every step is written out in Appendix~\ref{pf:rate}.
\end{proof}

Everything so far is exact. Here exactness ends: $m(\theta)$
solves a coupled integro-differential system driven by the nonlocal
generator of $(G,S)$, and it has no closed form. We therefore evaluate
the passage times on a surrogate (uniqueness:
Appendix~\ref{pf:surrogate}). Let $\tilde G$ be the unique Gaussian
diffusion with the gap's exact conditional mean and covariance structure
\eqref{eq:condmean}--\eqref{eq:moments},
\begin{equation}\label{eq:ou}
d\tilde G_t=-\alpha\,\tilde G_t\,dt+\sigma\,d\tilde Z_t,
\qquad
\sigma^{2}=\sigma_X^{2}+\sigma_M^{2}=2\alpha\,\sG^{2},
\end{equation}
and write $\tilde m(\theta):=\E_{-\theta}[T_{+\theta}]$, with
$T_y:=\inf\{t\ge0:\tilde G_t=y\}$, for \emph{its} passage
time, in closed form an $\erfi$ (Appendix~\ref{app:passage}). The object we
maximise is the surrogate rate
\begin{equation}\label{eq:ratesur}
\tilde R(\theta)=\frac{2\,(\theta-\phi)}{\tilde m(\theta)} .
\end{equation}
Two distinct symbols are needed: $m(\theta)$ is a
stationary mean over the fill chain of the true book, which begins each leg
\emph{outside} the band, at $\mp|G_F|$, whereas $\tilde m(\theta)$ is the passage
time of the surrogate from $-\theta$ to $+\theta$.
The \emph{construction} of $\tilde G$ assumes nothing: every parameter is an
identity of the model, \eqref{eq:condmean}--\eqref{eq:moments}. Its \emph{use},
however, is an approximation, and the three errors it introduces differ in
status. (a) Replacing $\E_{\pi_F}|G_F|$ by $\theta$ \emph{understates the
reward}, and the overshoot bound $\theta\le\E_{\pi_F}|G_F|<\theta+\delta$ of
Proposition~\ref{prop:rate} makes that error of relative order
$\delta/(\theta-\phi)$, rigorously: the denominator is the margin, not the
threshold, and at the optimum \eqref{eq:arrhenius} it equals $u^{*}\delta/\sG$.
(b) Replacing $\E_{\pi_F}[\phi_F]$ by $\phi=\delta/2$ \emph{overstates the reward},
by at most $\tfrac{\delta}{2}\,\pi_F(\text{open})$ per lot, small by the
occupancy of open books (Section~\ref{sec:remarks}) and pulling against
(a). (c) Replacing $m(\theta)$ by $\tilde m(\theta)$ \emph{replaces the timing},
and here the moments carry us no further: a first-passage time is not a function of
the conditional mean and stationary covariance, so matching those two moments does
not by itself control $\tilde m(\theta)-m(\theta)$. What we invoke instead is the
diffusion approximation of the limit theorems for microstructure point processes
\cite{bacry13b} (the same idealisation by which Avellaneda and Stoikov model the
mid as a Brownian motion \cite{as08}), under which, in the regime where a fill
is preceded by many small jumps, the timing error too is expected to be
of relative order $\delta/\theta$. (c) is a heuristic, not a theorem: a proof
that $|\tilde m(\theta)-m(\theta)|/m(\theta)=O(\delta/\theta)$ is left open, and the
optimal rate \eqref{eq:Rstar} below is exact only for the surrogate. This is where
we work: $\delta\ll\theta-\phi$, which forces $\delta\ll\theta$.

\paragraph{The band equation.} One fact about mean reversion finishes the
problem. For thresholds beyond about one standard deviation, the
surrogate's waiting time obeys Kramers' law \cite{kramers40}: \emph{the time to see the gap at $\pm\theta$ is
inversely proportional to the stationary probability of being there},
\begin{equation}\label{eq:kramers}
\tilde m(\theta)\;\propto\; e^{\theta^{2}/2\sG^{2}}
\end{equation}
up to an algebraic prefactor (the exact formula is in
Appendix~\ref{app:passage}). Keeping only the exponential factor,
maximising \eqref{eq:ratesur} is one line:
$\tfrac{d}{d\theta}\log(\theta-\phi)=\tfrac{d}{d\theta}\log \tilde m(\theta)$
gives $\tfrac{1}{\theta-\phi}=\tfrac{\theta}{\sG^{2}}$, i.e.
\begin{equation}\label{eq:arrhenius}
\boxed{\;\theta^{*}\big(\theta^{*}-\phi\big)=\frac{\sigma^{2}}{2\alpha}=\sG^{2}\;}
\end{equation}
The threshold times the option-value margin equals the stationary variance
of the gap. In units of $\sG$ the solution is
$u^{*}=\theta^{*}/\sG=\tfrac12\big(\gamma+\sqrt{\gamma^{2}+4}\,\big)$,
governed by one dimensionless number: the spread-to-dispersion ratio
$\gamma=\phi/\sG$; explicitly,
$\theta^{*}=\tfrac12\big(\phi+\sqrt{\phi^{2}+4\sG^{2}}\,\big)$. (The exact maximiser
$\theta_{\mathrm D}=\sG u_{\mathrm D}$ solves a
one-line equation in the Dawson function, Appendix~\ref{app:passage}. For
$\gamma\gtrsim0.4$ the two roots differ by at most six percent in threshold
and by $2.3$ percent in rate at worst, near
$\gamma\approx1.7$. Below $\gamma\approx0.28$ the exact optimum leaves
the Kramers regime altogether: $u_{\mathrm D}$ drops inside one standard deviation,
which $u^{*}\ge1$ never does, and only the Dawson root applies. Even at
$\gamma=0.05$, though, the flat maximum
caps the rate lost to \eqref{eq:arrhenius} near nine percent,
Figure~\ref{fig:threshold}. Section~\ref{sec:validate} benchmarks the
exact model at $\theta_{\mathrm D}$, the sharper root.)

\begin{figure}[!tb]
\centering
\includegraphics[width=\textwidth]{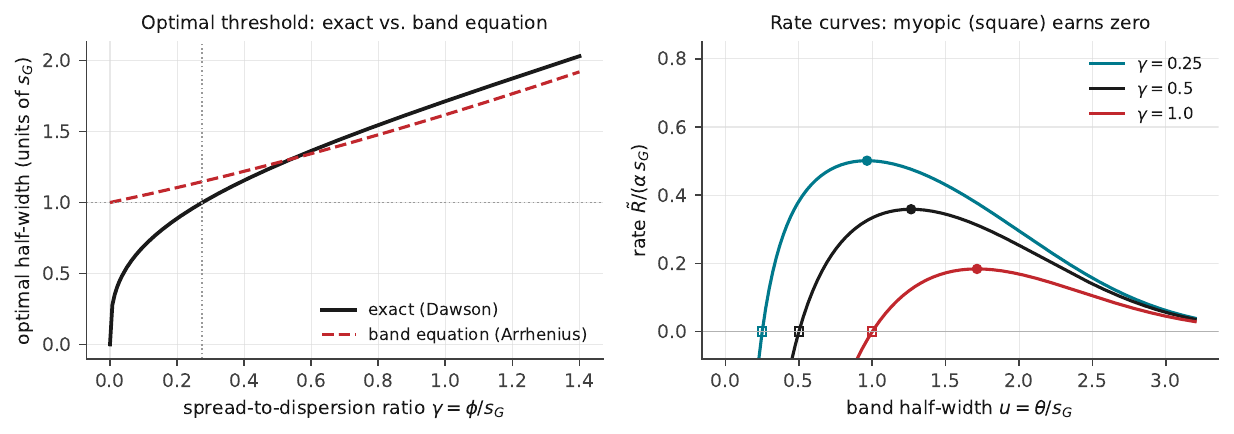}
\caption{Left: the two candidate half-widths against $\gamma$, the
band-equation root $u^{*}$ of
\eqref{eq:arrhenius} and the exact maximiser $u_{\mathrm D}$ of
Appendix~\ref{app:passage}; both are evaluated from the formulas, nothing
is estimated. The curves separate as $\gamma\to0$, where the exact optimum leaves
the Kramers regime $u\gtrsim1$ (dotted lines) and the Dawson root takes over. Right:
surrogate rate curves $\tilde R/(\alpha\sG)$ from \eqref{eq:ratesur};
squares mark the myopic threshold $u=\gamma$, which earns zero; dots mark
the optimum. The curves are flat at the top: a threshold error of
$\varepsilon$ costs only $O(\varepsilon^{2})$, and the simple root loses
$2.3$ percent of rate at worst for $\gamma\gtrsim0.4$.}
\label{fig:threshold}
\end{figure}

\paragraph{The optimal rate.} Substituting the optimum back,
\begin{equation}\label{eq:Rstar}
\boxed{\;R^{*}=\alpha\,\sG\,\sqrt{\tfrac{2}{\pi}}\;e^{-u^{*2}/2}\;}
\end{equation}
That is reversion speed, times gap dispersion, times the Gaussian weight of the
band's edge under the stationary law. (Equation \eqref{eq:Rstar} is exact
at the surrogate optimum: with $u_{\mathrm D}$ in place of $u^{*}$, the equality is
Proposition~\ref{prop:exact}. The rate lost by trading instead at the root $u^{*}$ of
\eqref{eq:arrhenius} is second order in the distance between the two
roots, since the rate is flat at its maximum (Lemma~\ref{lem:flat2}).) Patience is paid from the tails of the gap: hold out for a wider
traverse when the gap disperses more, settle for less when the spread
forces the band deep.

\paragraph{Zero for the impatient.} At $\theta=\phi$ each round trip nets
exactly zero while taking positive time, so $\tilde R(\phi)=0$: the myopic rule
(trade as soon as the gap covers the spread) earns \emph{nothing}.
The optimal band exceeds break-even by
$\theta^{*}-\phi=\sG^{2}/\theta^{*}$, and every unit of profit is the
option value of waiting for a deeper reversion. (In the exact model the
impatient rule keeps only the overshoot income net of the occasional
open-book cost, $\E_{\pi_F}[|G_F|-\phi_F]\in(-\delta/2,\delta)$ per lot:
the zero is the $\delta/\sG\to0$ limit of that income.)

\subsection{The surrogate optimum, tested against the exact model}\label{sec:validate}
One approximation stands between the formula and the model: the surrogate
of \eqref{eq:ou}, whose timing error was left heuristic. We test it where
it matters: at the optimum. Simulating the exact jump model of
Definitions~\ref{def:X}--\ref{def:M} and sweeping the band around the
surrogate optimum $\theta_{\mathrm D}$ of Appendix~\ref{app:passage},
Figure~\ref{fig:validate} reads off the realised rate and
its true maximiser. The sweep is over the band class: it tests where the
surrogate places the optimum and the rate it assigns there, not the
optimality of thresholds themselves, which Section~\ref{sec:band} left
as a conjecture on the exact process.

\begin{figure}[!tb]
\centering
\includegraphics[width=\textwidth]{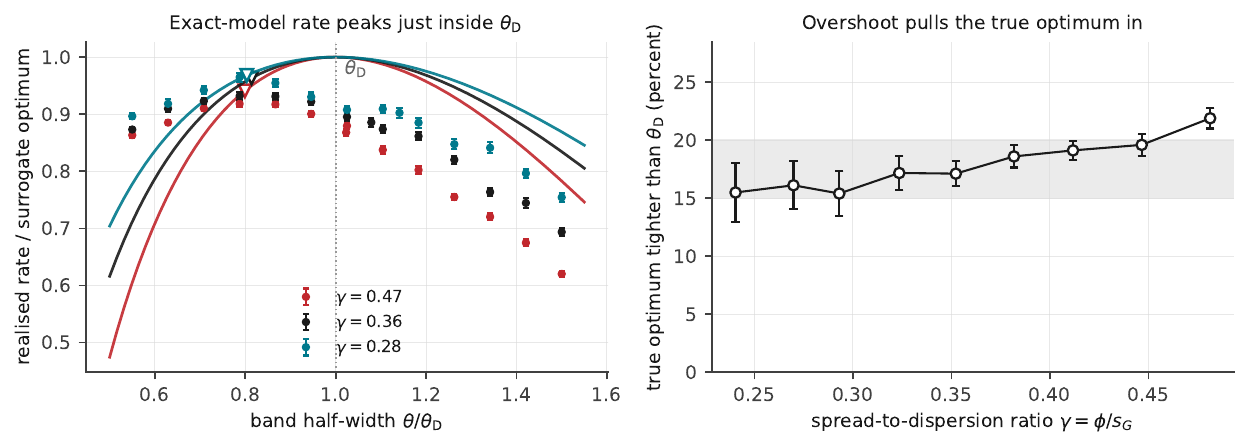}
\caption{The exact jump model of Definitions~\ref{def:X}--\ref{def:M} against
the surrogate optimum $\theta_{\mathrm D}$ of Appendix~\ref{app:passage}
(Monte Carlo; bands
are one standard error). Left: the realised rate against band half-width in
units of $\theta_{\mathrm D}$, normalised to the surrogate optimum, for three
representative $\gamma$; the curve is the surrogate rate, the points the
exact model, the marker ($\blacktriangledown$) its fitted peak. Right:
sweeping $\gamma$ through the ramp slopes of
Definitions~\ref{def:X}--\ref{def:M} with the baseline intensities held
fixed, the
peak sits inside
$\theta_{\mathrm D}$ by around a fifth throughout, edging up as $\gamma$ widens:
the pull of discrete-tick overshoot.}
\label{fig:validate}
\end{figure}

Two things come out, and both were expected. (a) The exact rate tracks the
surrogate over the whole band, and its true maximiser sits a little
\emph{inside} $\theta_{\mathrm D}$: by around a fifth across the swept range
of $\gamma$, edging up as the band widens. A jump carries the
gap past the threshold, so the reward arrives at a width slightly short of
the one the continuous surrogate would pick. (b) The cost of this misplacement
is small precisely because the rate is flat at its top
(Lemma~\ref{lem:flat2}): over the sweep, setting the band at $\theta_{\mathrm D}$ rather
than at the exact peak forfeits three to four percent of rate, and setting it at
the band-equation root $\theta^{*}$, which sits outside $\theta_{\mathrm D}$ at
these $\gamma$, five to six. The heuristic step
above is thus confirmed where it is used: the surrogate
misplaces the optimum by a controlled amount and misprices it by less. In
practice the flatness cuts the other way too: essentially all of the loss is
recovered anywhere in $[0.8,0.9]\,\theta_{\mathrm D}$. A practitioner may nudge the
band slightly inward, but either root is close enough that the choice barely
matters.

\paragraph{The band at work.} Figure~\ref{fig:band} closes the section on
one sample path of the exact model: the fills, the flips, the thresholds
breathing with the spread, and the wealth growing at the optimal rate.

\begin{figure}[!tb]
\centering
\includegraphics[width=\textwidth]{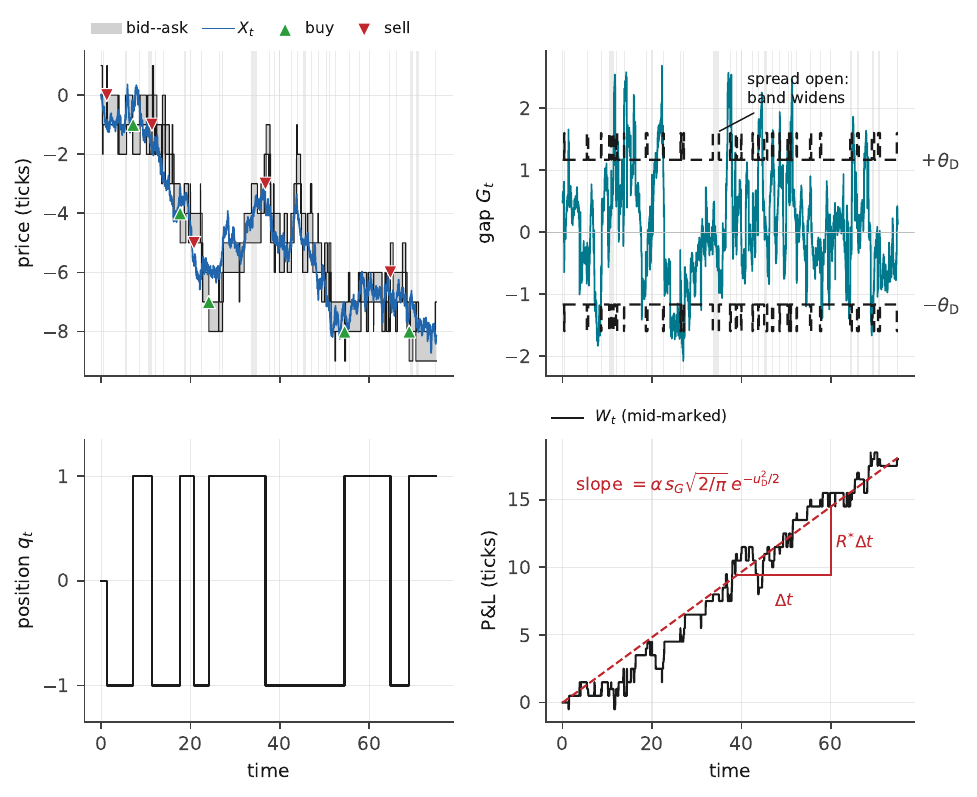}
\caption{The band at work: one sample path of the model of
Definitions~\ref{def:X}--\ref{def:M} (illustrative parameters;
nothing is calibrated), with the threshold set, in each spread state, at
the surrogate optimum $\theta_{\mathrm D}$ of Appendix~\ref{app:passage} for the \emph{realised}
half-spread: $\phi=\delta/2$ while the book is tight, $\phi=\delta$
during the open episodes (shaded), where the band visibly widens and
snaps back. Top to bottom: the bid--ask ribbon and the efficient price,
with buys ($\blacktriangle$) and sells ($\blacktriangledown$) executed
at the touch; the gap and the thresholds $\pm\theta_{\mathrm D}(S_t)$; the
position, one lot at the first entry, $\pm1$ ever after, never flat;
and the mid-marked wealth, whose growth tracks the dashed line of slope
$\tilde R(\theta_{\mathrm D})=\alpha\sG\sqrt{2/\pi}\,e^{-u_{\mathrm D}^{2}/2}$,
the exact optimal rate of the surrogate. The widenings are brief (their
occupancy is the $p$ of \eqref{eq:popen} and its gap-driven correction),
which is why the frozen-cost analysis at $\phi=\delta/2$ loses little.}
\label{fig:band}
\end{figure}

\section{Discussion}\label{sec:remarks}

\paragraph{All profit is option value.} The myopic rule, trade as
soon as the gap covers the spread, breaks even on the surrogate by
construction; whatever a strategy earns, it earns by declining trades
that rule would take.

\paragraph{One number governs everything.} In the gap's own units, $\sG$
for threshold and margin, $\alpha\sG$ for rate, the market enters through
one number: $\gamma=\phi/\sG$, the half-spread over the typical gap. Two
assets with the same $\gamma$ share the same band. When the gap barely exceeds the
spread ($\gamma$ large), the band sits deep and the rate is exponentially
small; when dispersion dominates ($\gamma$ small), the band-equation root
tightens toward one standard deviation, while the exact optimum drops
inside it below $\gamma\approx0.28$ (Section~\ref{sec:band}). Gaps are wide, and the band
with them, when the fundamental is noisy or the book corrects lazily:
$\sG^{2}=(\sigma_X^{2}+\sigma_M^{2})/2\alpha$ exactly, by \eqref{eq:moments},
with $\sigma_M^{2}$ a stationary mean that Remark~\ref{rem:bracket}
brackets.

\paragraph{The optimum is forgiving.} Because the optimum is an interior
maximum of a smooth curve, an error $\varepsilon$ in the threshold costs
only $O(\varepsilon^{2})$ in rate (Figure~\ref{fig:threshold}, right). The
strategy does not require knowing $\alpha$ and $\sG$ precisely, a
practical virtue, since both must be estimated.

\paragraph{Where the cost is paid.} The spread is wide
only while the book is open, and in the model of Definition~\ref{def:M} open books
are transitional by construction: entered by the open moves and left
by the closes; their occupancy is the $p$ of \eqref{eq:popen},
small in liquid large-tick books \cite{dayrirosenbaum15}. The band's fills therefore occur, as a rule, in tight books at
$\phi=\delta/2$, and freezing the cost there, as Section~\ref{sec:band}
does, misprices only the fills that land in open books: exactly
$\tfrac{\delta}{2}\,\pi_F(\text{open})$ per lot, by the two-valued cost
of Proposition~\ref{prop:rate}.
This localisation is also why venue rules that act only on
open books (price-improvement programmes, for instance) are
second-order for this strategy: they change the cost in a state the band
rarely uses.

\paragraph{What we left out.} (a) The gap is not directly observable: $X$
is latent, and a trading version of this paper must filter it from
quotes and trades. The ingredients for that attack are on the shelf.
The model of Section~\ref{sec:jumpmodel} is a state-space system: a
Brownian state observed through a point process whose intensities are
linear in the state. The implied filter is the point-process
one \cite{durbinkoopman12}, and both channels of the observation inform
the latent state. A jump of $M$ reweights the posterior on the gap by
the intensity of the observed move and shifts it by the move's known size. A quiet spell reweights it toward
small gaps, because the total event rate grows with dislocation. The
estimate of $X$ is therefore a function of the movement, and of the
non-movement, of $M$. Kalman recursions apply at the surrogate level,
where the counting observation aggregates to a Gaussian one.
Likelihood estimation for order flow driven by a latent Brownian
price, the structure here, is developed in \cite{delattre13};
and the efficient price is recovered from tick data under uncertainty
zones in \cite{robertrosenbaum11}. Estimation costs performance
relative to the full-information benchmark solved here; how much is a
filtering question we leave open. The parameters are the second
estimation problem: baselines and loadings are identified by standard
point-process likelihood \cite{daleyverejones03}, as in
\cite{rabechini19}, while the reduced pair $(\alpha,\sG)$ needs no
book-level data at all: $\alpha$ is the decay rate of the gap's
autocovariance \eqref{eq:moments}, $\sigma_X$ is delivered by
two-scale realised variance \cite{zhangmykland05}, and Hasbrouck's
decomposition recovers the permanent component from trades and quotes
\cite{hasbrouck95}. (b) Beyond estimation: liquidity providers may respond
strategically to being picked off at the gap's extremes, which turns
the problem into a game. (c) The switching problem can be attacked on
the point process of Section~\ref{sec:jumpmodel} itself, replacing the
surrogate passage times of \eqref{eq:ou} by the exact ones, with
overshoot then a feature, not an error term. (d) The objective
\eqref{eq:problem} is an expectation, not a utility, a modelling
choice. A concave utility, or the Sharpe variant that \cite{bertram10}
solves for the OU asset, would price the variance of the always-full
position, and the optimum changes shape: a flat zone reopens around
zero, interior positions can become optimal
(Proposition~\ref{prop:bangbang}), and the optimal thresholds
themselves move. Each of these is a natural
continuation; all start from the frame built here.

\section{Conclusion}\label{sec:conclusion}

Building the book from the bid shows that the spread merely duplicates
the parity of the mid (the parity lock), so the gap is the only
continuous state variable. In a mutually exciting book the gap's reversion at the book's leaning
rate $\alpha$ is a theorem, and the fundamental is proved to be the
long-run forecast of the mid. With waiting times evaluated on the exactly matched Gaussian
surrogate, the optimal trading of that reversion is a one-dimensional
switching problem with a closed-form answer: a symmetric
no-churn band with
$\theta^{*}(\theta^{*}-\phi)=\sG^{2}$ and rate
$\alpha\sG\sqrt{2/\pi}\,e^{-u^{*2}/2}$. Trading as soon as the gap covers
the spread earns exactly zero. The result is meant in the
spirit of Avellaneda and Stoikov \cite{as08}: a stylised model, an
explicit solution, and formulas simple enough to reason with.

Three assumptions carry the result: (a) the balanced response
\eqref{eq:balanced}, whose bite is bounded by the open-book
occupancy in \eqref{eq:aeff}; (b) the small trader, whose orders leave the
intensities \eqref{eq:intensity} unmoved; and (c) the observability of the
gap. Two questions remain open on the jump process itself: the timing
error of the surrogate, whose order $\delta/\theta$ is left heuristic,
and the optimality of the band class, stated here as a conjecture.
Section~\ref{sec:remarks} sketches the corresponding continuations:
estimation, the provider's game, the switching problem attacked on the
jump process directly, and the risk-adjusted criterion.

\appendix

\section{Proofs}\label{app:proofs}

The results quoted in the text are collected here with complete proofs, in order of
dependency: the parity lock first, since the state $(G,S)$ rests on it, then the
reversion and its stationary consequences, the wealth identity, the rate, and the
optimisation. Section~\ref{app:prelim} lists the facts they use.

\subsection{Preliminaries}\label{app:prelim}

Throughout, $\Delta Y_t:=Y_t-Y_{t-}$ denotes the jump at $t$ of a
right-continuous process $Y$, while the constants $\Delta_m$ are the mid
displacements of Proposition~\ref{prop:reversion} ($\pm\delta$ for
slides, $\pm\delta/2$ for opens and closes), so that
$\Delta M_t=\sum_m\Delta_m\,\Delta N^{m}_t$.

\paragraph{(P0) Existence and nonexplosion.} A process satisfying
Definitions~\ref{def:X}--\ref{def:M} exists and is nonexplosive.
Conditionally on the path of $X$, the intensities \eqref{eq:intensity} are
continuous, locally bounded functions of time between events, so the book is
built by iterating Ogata thinning from one event to the next
\cite{daleyverejones03}. Explosion is excluded by the drift bound proved in
\ref{pf:moments}: with $\tau_n$ the $n$-th event and $V=G^{2}$, the bound
$\mathcal LV\le-\alpha V+d\le d$ gives, by Dynkin's formula with the standard
double localisation and Fatou, $\E[V(G_{t\wedge\tau_n})]\le V(G_0)+d\,t$
uniformly in $n$; since the total intensity is affine,
$\sum_m\lambda^{m}\le\tilde c_0+\tilde c_1|G|$, it follows that
\[
\E\Big[\sum_m N^{m}_{t\wedge\tau_n}\Big]
=\E\!\int_0^{t\wedge\tau_n}\!\sum_m\lambda^{m}(u)\,du
\;\le\;\big(\tilde c_0+\tilde c_1\sqrt{V(G_0)+d\,t}\,\big)\,t\;<\;\infty
\]
uniformly in $n$: the expected number of events up to any fixed time is
finite, and $\tau_n\to\infty$ almost surely.

\paragraph{(P1) It\^o for a jump process minus a Brownian motion.} Let $G=M-X$
with $M$ of finite activity and pure jump and $X$ continuous. Then
$[G]_t=\sigma_X^{2}t+\sum_{u\le t}(\Delta M_u)^{2}$: $X$ contributes its quadratic
variation, $M$ its jumps, and the cross term vanishes because $X$ is continuous.
Consequently
\[
dG^{2}_t=2G_{t-}\,dG_t+\sigma_X^{2}\,dt+\sum_m\Delta_m^{2}\,dN^{m}_t .
\]

\paragraph{(P2) Foster--Lyapunov.} Let a Markov process have extended generator
$\mathcal L$, and suppose there exist $V\ge0$ with compact sublevel sets and
constants $c,d>0$ such that $\mathcal LV\le-cV+d$, \emph{and} that all compact
sets are petite. Then the process is positive Harris recurrent, has a unique
invariant law $\pi$, and $\pi(V)\le d/c$ \cite{meyntweedie93}. For $(G,S)$ the
petiteness is supplied by the noise the model already carries: $\sigma_X>0$
gives the gap a nondegenerate Brownian component between events, and the
positive baselines let the parity flip anywhere: from any initial state,
over a unit window, the event ``exactly $k\in\{0,1\}$ book events, with the
Brownian increment in a prescribed interval'' has probability bounded below,
locally uniformly in the initial state, and lands $(G,S)$ in any given open
set with positive probability. The process is therefore $\psi$-irreducible and
a $T$-process in the sense of \cite{meyntweedie93}, for which all compact sets
are petite.

\paragraph{(P3) Regenerative ratio (renewal--reward).} If visits to an
accessible atom split the path into i.i.d.\ cycles of finite mean length,
with rewards of bounded size and finite expected number per cycle, then
$T^{-1}\,\E[\sum_{\tau_k\le T}\mathcal R_k]\to
\E[\text{reward per cycle}]/\E[\text{cycle length}]$, almost surely and in
the mean \cite[Ch.~VI]{asmussen03}; equivalently, the ratio of stationary
means over the embedded chain.

\subsection{Proof of Fact~\ref{fact:parity} (the parity lock)}\label{pf:parity}

\begin{proof}
Measure the mid in half-ticks: $k_t:=2M_t/\delta=2B_t/\delta+S_t/\delta$. The
bid lives on the tick grid, so $2B_t/\delta$ is an even integer and the parity
of $k_t$ is that of $S_t/\delta$, for any spread on the grid: an odd
spread makes $k_t$ odd, an even one makes it even, which is
\eqref{eq:parity}. Within \eqref{eq:grid} only $\delta$ and $2\delta$
are available, so the map inverts: $S_t=\delta$ iff $k_t$ is odd and
$S_t=2\delta$ iff $k_t$ is even. The
correspondence holds at $t=0$ by construction and is preserved by every move of
\eqref{eq:mid}: a slide changes $k$ by $\pm2$ and leaves the spread; an open
changes $k$ by $\pm1$ and takes the spread from $\delta$ to $2\delta$; a close
is its mirror image. In each of the six moves the parity of $k$ and the value
of $S$ change together, so the lock holds for all $t$: the spread is a
deterministic function of the mid, and $(G,S)$ is a Markov process in the
single continuous coordinate $G$ together with a parity bit. \emph{No dynamics
for $S$ needs to be specified: it is read off $M$.}
\end{proof}

\subsection{Proof of Proposition~\ref{prop:reversion}: the reversion
identity}\label{pf:drift}

\begin{proof}[Proof of \eqref{eq:drift}]
\emph{Step 1 (the compensator of the mid).} By \eqref{eq:mid} and the
compensator property of \eqref{eq:poisson} (no two of the six counting
processes jump together, so the compensator of a weighted sum of them is the
weighted sum of their intensities \cite{daleyverejones03}),
$\E[dM_t\mid\mathcal F_{t-}]$ is the sum of the six intensities weighted by their
moves: $\pm\delta$ for slides, $\pm\delta/2$ for opens and for closes.

\emph{Step 2 (a tight book).} On $\{S_{t-}=\delta\}$ only slides and opens have
positive intensity, and their weighted difference is
\[
\delta\Big[\mu_s+\tfrac{2\alpha_s}{\delta}G^{-}\Big]
-\delta\Big[\mu_s+\tfrac{2\alpha_s}{\delta}G^{+}\Big]
+\tfrac{\delta}{2}\Big[\mu_o+\tfrac{2\alpha_o}{\delta}G^{-}\Big]
-\tfrac{\delta}{2}\Big[\mu_o+\tfrac{2\alpha_o}{\delta}G^{+}\Big].
\]
The baselines cancel pairwise, and what remains is
$(2\alpha_s+\alpha_o)(G^{-}-G^{+})$; since $G^{+}-G^{-}=G$ identically
(the positive parts are there only to keep each intensity nonnegative, and they
cancel in every up-minus-down difference), this equals
$-(2\alpha_s+\alpha_o)\,G_{t-}$.

\emph{Step 3 (an open book).} On $\{S_{t-}=2\delta\}$ only closes are active,
and the same computation gives
$\tfrac{\delta}{2}\cdot\tfrac{2\alpha_c}{\delta}\,(G^{-}-G^{+})
=-\alpha_c\,G_{t-}$.

\emph{Step 4 (the balanced response).} Condition \eqref{eq:balanced} reads
$2\alpha_s+\alpha_o=\alpha_c=:\alpha$, so Steps 2 and 3 agree:
$\E[dM_t\mid\mathcal F_{t-}]=-\alpha\,G_{t-}\,dt$ in either parity, with no
expansion and no approximation.

\emph{Step 5 (the gap).} $X$ is a martingale by Definition~\ref{def:X}, so
$\E[dX_t\mid\mathcal F_{t-}]=0$ and
$\E[dG_t\mid\mathcal F_{t-}]=-\alpha G_{t-}\,dt$, which
is \eqref{eq:drift}. No dynamics for $G$ was assumed: the reversion is a
consequence of \eqref{eq:intensity} and \eqref{eq:balanced}.
\end{proof}

\subsection{Proof of Proposition~\ref{prop:reversion}: stationarity and the
moments}\label{pf:moments}

\begin{proof}[Proof of ergodicity, \eqref{eq:condmean}, and \eqref{eq:moments}]
\emph{Step 1 (ergodicity).} Take $V(G,S)=G^{2}$. By (P1) and Step 5 of
\ref{pf:drift}, the generator of $(G,S)$ satisfies
$\mathcal LV=-2\alpha G^{2}+\sigma_X^{2}+\sum_m\lambda^{m}\Delta_m^{2}$. Each
intensity in \eqref{eq:intensity} is affine in $|G|$ and each $|\Delta_m|\le\delta$,
so $\sum_m\lambda^{m}\Delta_m^{2}\le c_0+c_1|G|$ with $c_0,c_1$ depending only on
the parameters. Using $c_1|G|\le\alpha G^{2}+c_1^{2}/4\alpha$,
\[
\mathcal LV\;\le\;-\alpha V+\Big(c_0+\sigma_X^{2}+\tfrac{c_1^{2}}{4\alpha}\Big),
\]
which is (P2) with $c=\alpha$; compact sets are petite by the noise argument
recorded there. Hence $(G,S)$ is positive Harris recurrent with a
unique invariant law $\pi$ and $\E_\pi[G^{2}]<\infty$. Every ``stationary''
quantity below refers to $\pi$.

\emph{Step 2 (the mean is zero).} Taking $\E_\pi$ in \eqref{eq:drift} gives
$0=-\alpha\,\E_\pi[G]$, so $\E_\pi[G]=0$, the first identity of
\eqref{eq:moments}; it also follows from mirror symmetry alone.

\emph{Step 3 (the variance).} By (P1) and stationarity (all
local-martingale terms are true martingales here, their integrands bounded in
$L^{1}$ by $\E_\pi[G^{2}]<\infty$),
\[
0=\E_\pi\big[dG^{2}\big]/dt
=-2\alpha\,\E_\pi[G^{2}]+\sigma_X^{2}
+\E_\pi\Big[\sum_m\lambda^{m}(t)\,\Delta_m^{2}\Big] .
\]
Writing $\sigma_M^{2}$ for the last expectation, the stationary mean of the
book's jump-variance rate, gives
$\sG^{2}=\E_\pi[G^{2}]=(\sigma_X^{2}+\sigma_M^{2})/2\alpha$, the middle identity
of \eqref{eq:moments}. This is the one place where the nonlinearity of
\eqref{eq:intensity} survives: it enters $\sigma_M^{2}$ through $\E_\pi|G|$, and
nowhere else.

\emph{Step 4 (the conditional mean at every horizon).} By the drift bound of
Step 1 and Gronwall,
$\E[G_{t+r}^{2}\mid\mathcal F_t]\le G_t^{2}+d/\alpha$ for all $r\ge0$, so the
compensated integrals below are true martingales and Fubini applies. Put
$f(h):=\E[G_{t+h}\mid\mathcal F_t]$. Integrating \eqref{eq:drift} and applying
Fubini, $f(h)=G_t-\alpha\int_0^{h}f(r)\,dr$; hence $f'=-\alpha f$ with $f(0)=G_t$,
so $f(h)=G_t\,e^{-\alpha h}$: this is \eqref{eq:condmean}. Only the linear
compensator is used, not the Markov
property and no stationarity: the identity holds from every state.

\emph{Step 5 (the autocovariance).} By \eqref{eq:condmean} and the tower property,
$\mathrm{Cov}_\pi(G_t,G_{t+h})
=\E_\pi\big[G_t\,\E[G_{t+h}\mid\mathcal F_t]\big]
=e^{-\alpha h}\,\E_\pi[G_t^{2}]=\sG^{2}e^{-\alpha h}$, which completes
\eqref{eq:moments}.
\end{proof}

\subsection{A lemma on the two markings}\label{pf:wealth}

\begin{lemma}\label{lem:markings}
Under \eqref{eq:wealth}, $W_t-W^{X}_t=q_tG_t$ for all $t$; in particular
$\sup_t\E\,|W_t-W^{X}_t|<\infty$, and the two markings define the same
long-run rate.
\end{lemma}

\begin{proof}
\emph{Step 1 (the difference is an exact differential).} By \eqref{eq:wealth},
\[
dW_t-dW^{X}_t=q_{t-}\big(dM_t-dX_t\big)+G_t\,dq_t
=q_{t-}\,dG_t+G_t\,dq_t .
\]
Since $q$ is piecewise constant and of finite variation while $G$ is a
semimartingale, integration by parts gives
$d(qG)_t=q_{t-}dG_t+G_{t-}dq_t+d[q,G]_t$ with
$[q,G]_t=\sum_{s\le t}\Delta q_s\,\Delta G_s$,
and $G_{t-}\Delta q_t+\Delta q_t\Delta G_t=G_t\,\Delta q_t$. Hence
$dW_t-dW^{X}_t=d(q_tG_t)$ and, starting flat, $W_t-W^{X}_t=q_tG_t$.

The identity requires the trade term in $dW^{X}$ to carry the
\emph{post-jump} gap $G_t\,dq_t$, not $G_{t-}\,dq_t$, the execution
convention of the model; with $G_{t-}$ the identity would carry the extra term
$\Delta q\,\Delta G$.

\emph{Step 2 (boundedness).} $|q|\le1$ and $\sup_t\E|G_t|<\infty$ by the drift
bound in Step 4 of \ref{pf:moments}, so $\E|W_t-W^{X}_t|$ is bounded uniformly
in $t$ and the two markings define the same long-run rate.
\end{proof}

\subsection{Proof of Proposition~\ref{prop:bangbang}}\label{pf:bangbang}

\begin{proof}
\emph{Step 1 (the layers are admissible and reconstruct the position).}
Each $a^{u}$, as defined in the statement, is a deterministic function of
$q_t$, hence adapted; it is
right-continuous and piecewise constant because $q$ is; it trades only when
$q$ trades, with $|\Delta a^{u}|\le2$ per trade, so its expected number of
trades on bounded intervals is finite and $a^{u}$ is admissible. The layers
reconstruct the position exactly: for every $t$,
\[
\int_0^1 a^{u}_t\,du
\;=\;\int_0^1 \mathbf 1\{q_t\ge u\}\,du-\int_0^1\mathbf 1\{q_t\le -u\}\,du
\;=\;q_t^{+}-q_t^{-}\;=\;q_t .
\]

\emph{Step 2 (gains are linear, so they average exactly).} By nonexplosion
(P0) the mid has finitely many jumps on $[0,T]$, so the gains integral is
the finite sum $\int_0^T q_{t-}\,dM_t=\sum_i q_{\tau_i-}\,\Delta
M_{\tau_i}$ over the jump times of $M$. Substituting the layer identity of
Step 1 at each $\tau_i-$ and exchanging the finite sum with the integral in
$u$,
\[
\int_0^T q_{t-}\,dM_t \;=\;\int_0^1\Big(\int_0^T a^{u}_{t-}\,dM_t\Big)du .
\]

\emph{Step 3 (proportional costs obey an exact coarea identity).} The map
$x\mapsto\mathbf 1\{x\ge u\}-\mathbf 1\{x\le-u\}$ is nondecreasing with
unit jumps at $x=-u$ and $x=u$. Hence, for a trade of $q$ at time $\tau$
from $a:=q_{\tau-}$ to $b:=q_{\tau}$,
\[
|\Delta a^{u}_\tau|
=\mathbf 1\{u\in(a\wedge b,\,a\vee b]\}
+\mathbf 1\{-u\in[a\wedge b,\,a\vee b)\},
\]
and integrating in $u$ the two indicator sets tile the positive and
negative parts of the crossed interval:
$\int_0^1|\Delta a^{u}_\tau|\,du=|b-a|=|\Delta q_\tau|$. (For instance, a
trade from $+0.3$ to $-0.6$ flips the layers $u\le0.3$ from $+1$ to $-1$
and sells the layers $0.3<u\le0.6$ from $0$ to $-1$:
$2\times0.3+1\times0.3=0.9$.) The half-spread $\phi_\tau$ at the trade
time is common to all layers, so
$\phi_\tau|\Delta q_\tau|=\int_0^1\phi_\tau|\Delta a^{u}_\tau|\,du$, and
summing over the finitely many trades,
\[
\int_0^T \phi_t\,|dq_t| \;=\;\int_0^1\Big(\int_0^T\phi_t\,|da^{u}_t|\Big)du .
\]

\emph{Step 4 (the wealth identity and the pre-limit bound).} Subtracting
Step 3 from Step 2 gives \eqref{eq:layer} pathwise. The family is jointly
measurable, and integrable uniformly in $u$: the gains are dominated by
$\delta$ times the number of book events, of finite mean by (P0), and the
costs by $\delta$ times twice the number of trades of $q$, of finite mean
by admissibility. Fubini then gives
$\E[W_T(q)]=\int_0^1\E[W_T(a^{u})]\,du\le\sup_{u}\E[W_T(a^{u})]$ for
every $T$: an average cannot beat its best layer.

\emph{Step 5 (the equality of values).} One inequality in \eqref{eq:supeq}
is inclusion: a three-valued strategy is an admissible interval-valued one.
For the other, enlarge the filtration by a draw $U$, uniform on $(0,1]$,
taken at time zero independently of the book (the independence leaves
every intensity and every martingale statement unchanged, and
admissibility is understood throughout to permit such a draw), and set
$\hat a:=a^{U}$, a single admissible strategy with values in
$\{-1,0,+1\}$. By independence, Fubini
and the integrability of Step 4,
\[
\E[W_T(\hat a)]\;=\;\int_0^1\E[W_T(a^{u})]\,du\;=\;\E[W_T(q)]
\qquad\text{for every }T ,
\]
so $\hat a$ and $q$ share the whole function $T\mapsto\E[W_T]$, hence the
same lower-limit rate, and \eqref{eq:supeq} follows.
\end{proof}

\begin{remark}\label{rem:norand}
Randomisation is a convenience, not a loophole. On the market filtration
alone, the moment bound of Appendix~\ref{pf:moments} gives
$K:=\sup_{t\ge0}\E|G_t|<\infty$ and, by \eqref{eq:drift},
$T^{-1}\E[W_T(a^{u})]\le\alpha K$ uniformly in $T\ge1$ and $u$; reverse
Fatou applied to \eqref{eq:layer} then yields
$\limsup_T T^{-1}\E[W_T(q)]\le\sup_u\limsup_T T^{-1}\E[W_T(a^{u})]$:
no interval-valued strategy beats the three-valued value in the upper-limit
sense either, and the optimal band attains its rate as a true limit
(Appendix~\ref{pf:rate}).
\end{remark}

\subsection{Proof of Proposition~\ref{prop:rate}}\label{pf:rate}

\begin{proof}
\emph{Step 1 (the two markings agree).} By \eqref{eq:wealth} and
Appendix~\ref{pf:wealth}, $W_t-W^{X}_t=q_tG_t$. Since $|q|\le1$ and
$\sup_t\E|G_t|<\infty$ (Step 4 of \ref{pf:moments} bounds
$\E[G_t^{2}]$ from any start), the difference is bounded in $L^{1}$
uniformly in $t$, so $\E[W_T]/T$ and $\E[W^{X}_T]/T$ have the same limit.

\emph{Step 2 (the inventory term vanishes).} $q_{t-}$ is predictable and
bounded and $X$ is a square-integrable martingale, so
$\int_0^{T}q_{s-}\,dX_s$ is a martingale of zero mean. Hence
$\E[W^{X}_T]=\E\big[-\int_0^{T}G_t\,dq_t-\int_0^{T}\phi_t\,|dq_t|\big]$: the rate
depends on the strategy only through the state at its trades.

\emph{Step 3 (the mean time between fills is finite, uniformly).} Under the
band \eqref{eq:band} the inventory is $\pm1$ after the first fill, and the fill
times $\tau_1<\tau_2<\cdots$ are the successive times the gap reaches the far
threshold; the post-fill gap satisfies $|G_{\tau_k}|\in[\theta,\theta+\delta)$
by Step 7 below. Positive recurrence of the fill chain requires the mean time
to \emph{leave} a bounded region, an outward hitting, which Harris
recurrence of $(G,S)$ does not by itself deliver. It does here. Fix $h>0$ and
let $\Lambda$ bound the total intensity $\sum_m\lambda^{m}$ on
$\{|G|\le2\theta+2\delta\}$, finite because each intensity is affine in $|G|$.
From any state with $|G|\le\theta+\delta$, consider the event $A$: no book
event on $[0,h]$, and the Brownian path stays in a tube that carries the gap
monotonically enough across the far threshold within $[0,h]$ while never
leaving $\{|G|\le2\theta+2\delta\}$. Conditionally on any such Brownian path,
the no-event probability is at least $e^{-\Lambda h}$, because until the first
event the gap moves by the Brownian alone and the tube confines it where the
intensities are bounded by $\Lambda$; and the tube itself has probability
bounded below, uniformly over the compact set of starting points, by the
standard tube estimate for Brownian motion. Hence
$\mathbb P(\text{fill within }h)\ge e^{-\Lambda h}\,p_{\mathrm{tube}}=:p>0$
uniformly, and by the strong Markov property
$\mathbb P(\tau_{k+1}-\tau_k>jh\mid\mathcal F_{\tau_k})\le(1-p)^{j}$, so
$\E[\tau_{k+1}-\tau_k\mid\mathcal F_{\tau_k}]\le h/p$ uniformly over post-fill
states.

\emph{Step 4 (an accessible atom, hence positive Harris).} By the strong
Markov property of $(G,S)$, the post-fill states
$(G_{\tau_k},S_{\tau_k},q_{\tau_k})$ form a Markov chain on the compact set
$\{\theta\le|g|<\theta+\delta\}\times\{\delta,2\delta\}\times\{\pm1\}$. On the
event $A$ of Step 3 the crossing is diffusive (no book event fires), so
the arrival gap is \emph{exactly} $\pm\theta$ and the parity arrives frozen at
its post-fill value; a variant of $A$ with exactly one parity-flipping event
early in the window (an open if the book is tight, total rate at least
$2\mu_o$ on the tube; a close if open, at least $2\mu_c$) and no other
events delivers the crossing with the \emph{opposite} parity, the tube
absorbing the $\delta/2$ kick. Both variants have probability bounded below,
uniformly over post-fill states. The chain therefore reaches each of the four
states $a=(\pm\theta,\sigma,\mp1)$, $\sigma\in\{\delta,2\delta\}$, in one step
with probability at least $\varepsilon>0$ from everywhere: these are
accessible atoms. Visits to a fixed atom then split the path into i.i.d.\
cycles whose mean length is finite (a geometric number of fills, each of
conditional mean at most $h/p$ by Step 3), so the chain is positive Harris
with an invariant law $\pi_F$ \cite[Ch.~VI]{asmussen03}. Write
$m(\theta):=\E_{\pi_F}[\tau_{k+1}-\tau_k]\le h/p<\infty$ for the stationary
mean time between opposite fills; mirror symmetry makes it the same for both
legs. \emph{This is the definition of $m(\theta)$ in \eqref{eq:rate}; it is
not the passage time of any diffusion.}

\emph{Step 5 (the reward at a fill).} At a sell fill the gap is
$G_F\ge\theta$, the half-spread is $\phi_F\in\{\delta/2,\delta\}$ read off the
parity at the fill, and the inventory flips from $+1$ to $-1$: $dq=-2$ and, by
\eqref{eq:wealth}, the $X$-marked cash flow is
$-G_F\cdot(-2)-\phi_F\cdot2=2(G_F-\phi_F)$. A buy fill at $-G_F'$ gives
$2(G_F'-\phi_F')$. A cycle contains one of each, so its reward has mean
$4\,\E_{\pi_F}[|G_F|-\phi_F]$ and its length mean $2m(\theta)$.

\emph{Step 6 (the regenerative ratio).} Rewards per fill are bounded by
$2(\theta+\delta)$ and cycles have finite mean length, so (P3), applied at the
atom of Step 4, gives
$R(\theta)=4\,\E_{\pi_F}[|G_F|-\phi_F]\big/\big(2m(\theta)\big)$, which is
\eqref{eq:rate}. The convergence is almost sure, so the pathwise rate equals this
stationary ratio, the ergodic identity, and the $\liminf$ in
\eqref{eq:problem} is attained from any start.

\emph{Step 7 (the overshoot bound).} The gap leaves $(-\theta,\theta)$ either
continuously, in which case $|G_F|=\theta$, or at a jump of $M$, whose size is at
most $\delta$ by \eqref{eq:mid}. Hence $\theta\le|G_F|<\theta+\delta$ pathwise,
and the same bounds hold under $\pi_F$.
\end{proof}

\subsection{Flatness is dominated}\label{pf:flat}

We record what is proved and what is cited. Let $q$ be admissible and consider an
\emph{excursion to flat}: the trader holds $q\in\{-1,1\}$, unwinds to $0$ at a
time $\varsigma_a$ when the gap is $G_a$, and re-enters the same $q$ at
$\varsigma_b>\varsigma_a$ when the gap is $G_b$.

\emph{Step 1 (the cost of the excursion).} By \eqref{eq:wealth} the $X$-marked
cash flows of the two trades are $q\,G_a-\phi$ at $\varsigma_a$ (where
$dq=-q$) and $-q\,G_b-\phi$ at $\varsigma_b$ (where $dq=+q$). Between them the
inventory is flat, so the only other term, $\int q\,dX$, vanishes pathwise.
Holding through the interval adds only the inventory martingale
$q\,(X_{\varsigma_b}-X_{\varsigma_a})$, of zero mean at the stopping times. In
expectation the excursion is therefore worth
\[
q\,(G_a-G_b)-2\phi\qquad\text{relative to holding,}
\]
the half-spread being paid on each side (we write the frozen tight-book
$\phi$; carrying the realised $\phi_F$ changes nothing in the comparison).

\emph{Step 2 (what Step 1 does and does not settle).} Two consequences are
immediate: an instantaneous round trip ($\varsigma_b=\varsigma_a$) pays
$2\phi$ for a traverse of zero and is never optimal, and the excursion is
\emph{ex post} profitable exactly when $q(G_a-G_b)>2\phi$. What Step 1 does
not settle is whether such an excursion can be chosen \emph{ex ante}: the
trader commits at one stopping time, the payoff depends on the gap at a later
one, and a pathwise inequality about a realised swing is not a statement about
conditional expectations. It does not by itself dominate holding.

\emph{Step 3 (what is cited).} That the optimum among \emph{all} admissible
strategies is a two-threshold flip rule, in particular that no flat zone is
opened between the legs, is the content of the verification theory for switching problems on a
one-dimensional Markov state with a monotone trip payoff: proved under
discounting in \cite{zhangzhang08,zervos13,leungli16}, carried to the
long-run average by vanishing-discount arguments \cite{abg12}. Mirror symmetry of
Definition~\ref{def:M} and of \eqref{eq:wealth} then makes the two thresholds
symmetric. We do not reprove those results, and we do not claim an elementary
argument in their place.

\subsection{The surrogate is the unique moment match}\label{pf:surrogate}

\begin{proof}[Proof of uniqueness]
Let $d\tilde G=-a\tilde G\,dt+s\,d\tilde Z$ with $a,s>0$; its conditional
mean is $\tilde G_t e^{-ah}$ and its stationary variance $s^{2}/2a$. Matching
the conditional mean \eqref{eq:condmean} at every horizon forces $a=\alpha$;
matching the stationary variance then forces
$s^{2}=2\alpha\sG^{2}=\sigma_X^{2}+\sigma_M^{2}$ by \eqref{eq:moments}:
these are the coefficients of \eqref{eq:ou}, and no other pair reproduces
both moments. What is \emph{not} matched is the exit mechanism: $\tilde G$
has continuous paths and exits an interval exactly at its boundary, while $G$
may overshoot by a jump of size at most $\delta$. Nothing in the two matched
moments constrains the jump structure; the surrogate replaces it, and this
replacement, together with the tight-book cost convention, is where
approximation enters: Section~\ref{sec:band} separates the proved
reward-side size, relative order $\delta/(\theta-\phi)$, from the heuristic
timing side, $\delta/\theta$.
\end{proof}

\section{Closed forms for the surrogate}\label{app:passage}

This appendix records the closed forms, for the surrogate \eqref{eq:ou}, of the
quantities used in Section~\ref{sec:band}. Throughout, $\mathcal N$ denotes the
standard normal distribution function; $\phi$ remains the half-spread.

\begin{proposition}[Passage time]\label{prop:passage}
For \eqref{eq:ou} with $\sG^{2}=\sigma^{2}/2\alpha$ and $u=\theta/\sG$,
\[
\tilde m(\theta)=\E_{-\theta}[T_{+\theta}]
=\frac{\sqrt{2\pi}}{\alpha}\int_{0}^{u}e^{x^{2}/2}\,dx
=\frac{\pi}{\alpha}\,\erfi\!\Big(\frac{u}{\sqrt2}\Big)<\infty .
\]
\end{proposition}

\begin{proof}
The scale and speed densities of \eqref{eq:ou} are
$\mathcal S(y)=e^{y^{2}/2\sG^{2}}$ and
$\mathcal M(y)=\tfrac{2}{\sigma^{2}}e^{-y^{2}/2\sG^{2}}$, both boundaries
natural, so passage times between interior points are finite and the classical
formula
$\E_a[T_b]=\int_a^{b}\mathcal S(y)\big(\int_{-\infty}^{y}\mathcal M(z)\,dz\big)dy$
applies \cite[Ch.~15]{karlintaylor81}. The inner integral is
$\tfrac{2\sG\sqrt{2\pi}}{\sigma^{2}}\,\mathcal N(y/\sG)$; substituting
$y=\sG x$ and using $2\sG^{2}/\sigma^{2}=1/\alpha$,
\[
\E_{-\theta}[T_{+\theta}]
=\frac{\sqrt{2\pi}}{\alpha}\int_{-u}^{u}e^{x^{2}/2}\,\mathcal N(x)\,dx
=\frac{\sqrt{2\pi}}{\alpha}\int_{0}^{u}e^{x^{2}/2}\,dx ,
\]
the second equality because $e^{x^{2}/2}$ is even and
$\mathcal N(x)+\mathcal N(-x)=1$. The substitution $x=\sqrt2\,t$ names the
integral $\sqrt{\pi/2}\,\erfi(u/\sqrt2)$, and the constants collapse to
$\pi/\alpha$.
\end{proof}

\begin{proposition}[The surrogate optimum]\label{prop:exact}
Let $\gamma=\phi/\sG>0$. The surrogate rate \eqref{eq:ratesur} has a unique
interior maximiser $\theta_{\mathrm D}=\sG u_{\mathrm D}$, where $u_{\mathrm D}>\gamma$ is the unique root
of
\begin{equation}\label{eq:dawson}
u-\gamma=\sqrt2\,\daw\!\Big(\frac{u}{\sqrt2}\Big),
\qquad
\daw(z)=e^{-z^{2}}\!\int_0^{z}e^{t^{2}}dt ,
\end{equation}
and $\tilde R(\theta_{\mathrm D})=\alpha\sG\sqrt{2/\pi}\,e^{-u_{\mathrm D}^{2}/2}$:
formula \eqref{eq:Rstar}, exact at this root. The band equation
\eqref{eq:arrhenius} and its root $u^{*}$ are the $u\gtrsim1$ asymptote, since
$\sqrt2\,\daw(u/\sqrt2)=1/u+O(u^{-3})$.
\end{proposition}

\begin{proof}
By \eqref{eq:ratesur} and Proposition~\ref{prop:passage},
$\tilde R(u)=\tfrac{2\alpha\sG}{\pi}\,(u-\gamma)/\erfi(u/\sqrt2)$ for
$u>\gamma$; since $\tfrac{d}{du}\erfi(u/\sqrt2)=\sqrt{2/\pi}\,e^{u^{2}/2}$
and $\daw(z)=\tfrac{\sqrt\pi}{2}e^{-z^{2}}\erfi(z)$, the first-order
condition $\tilde R'(u)=0$ rearranges to \eqref{eq:dawson}.

\emph{Step 1 (a unique root).} Put $h(u)=u-\gamma-\sqrt2\,\daw(u/\sqrt2)$. Then
$h(\gamma)=-\sqrt2\,\daw(\gamma/\sqrt2)<0$, and $h(u)\to\infty$ because
$\daw(z)\to0$. By the chain rule $h'(u)=1-\daw'(u/\sqrt2)$; since
$\daw'(z)=1-2z\,\daw(z)<1$ for $z>0$, $h'>0$. So $h$ is strictly increasing and
has exactly one root $u_{\mathrm D}>\gamma$.

\emph{Step 2 (the root is the maximum).} $\tilde R$ is continuous on
$[\gamma,\infty)$ with $\tilde R(\gamma)=0$, $\tilde R>0$ on $(\gamma,\infty)$,
and $\tilde R(u)\to0$ as $u\to\infty$ because $\erfi(u/\sqrt2)$ grows like
$e^{u^{2}/2}$. A positive continuous function vanishing at both ends of an
interval attains an interior maximum, which is a stationary point; by Step~1
there is exactly one. At the root the first-order condition reads
$\erfi(u_{\mathrm D}/\sqrt2)=(u_{\mathrm D}-\gamma)\sqrt{2/\pi}\,e^{u_{\mathrm D}^{2}/2}$, and substituting
it back the factor $(u_{\mathrm D}-\gamma)$ cancels:
$\tilde R(u_{\mathrm D})=\tfrac{2\alpha\sG}{\pi}\big(\sqrt{2/\pi}\,e^{u_{\mathrm D}^{2}/2}\big)^{-1}
=\alpha\sG\sqrt{2/\pi}\;e^{-u_{\mathrm D}^{2}/2}$, which is \eqref{eq:Rstar} at
$u_{\mathrm D}$.

\emph{Step 3 (the asymptote).} From
$\daw(z)=\tfrac{1}{2z}+\tfrac{1}{4z^{3}}+O(z^{-5})$ one gets
$\sqrt2\,\daw(u/\sqrt2)=\tfrac1u+O(u^{-3})$, so \eqref{eq:dawson} becomes
$u-\gamma=1/u$, which is \eqref{eq:arrhenius}. Kramers' law \eqref{eq:kramers} is
the same asymptote read on the passage time itself:
$\tilde m(\theta)=\tfrac{\pi}{\alpha}\,\erfi(u/\sqrt2)
\sim\sqrt{2\pi}\,\tfrac{\sG}{\alpha\theta}\,e^{\theta^{2}/2\sG^{2}}$.
\end{proof}

\begin{lemma}[Second-order flatness]\label{lem:flat2}
$\tilde R$ is twice continuously differentiable on $(\phi,\infty)$ with
$\tilde R''(\theta_{\mathrm D})<0$, so
$\tilde R(\theta_{\mathrm D}+\varepsilon)
=\tilde R(\theta_{\mathrm D})+\tfrac12\tilde R''(\theta_{\mathrm D})\varepsilon^{2}
+o(\varepsilon^{2})$: a relative error $\varepsilon$ in the threshold costs
$O(\varepsilon^{2})$ in rate.
\end{lemma}

\begin{proof}
Differentiability is clear, $\erfi$ being entire and nonvanishing on
$(\gamma,\infty)$, and $\theta_{\mathrm D}$ is an interior maximum, so
$\tilde R'(\theta_{\mathrm D})=0$ and Taylor's theorem applies; the strict sign is
transversality. In the variable $u$ the numerator of $\tilde R'$ is
$N(u)=\erfi(u/\sqrt2)-(u-\gamma)\sqrt{2/\pi}\,e^{u^{2}/2}$, and at the root
$N'(u_{\mathrm D})=-(u_{\mathrm D}-\gamma)\,u_{\mathrm D}\sqrt{2/\pi}\,e^{u_{\mathrm D}^{2}/2}<0$: $\tilde R'$
crosses zero strictly downward.
\end{proof}


\begin{thebibliography}{99}\small

\bibitem{ahuja17} S.~Ahuja, G.~Papanicolaou, W.~Ren, and T.-W.~Yang.
Limit order trading with a mean reverting reference price.
\emph{Risk and Decision Analysis}, 6(2):121--136, 2017.



\bibitem{almgren01} R.~Almgren and N.~Chriss.
Optimal execution of portfolio transactions.
\emph{Journal of Risk}, 3(2):5--39, 2001.

\bibitem{andersen03} T.~G.~Andersen, T.~Bollerslev, F.~X.~Diebold, and
P.~Labys. Modeling and forecasting realized volatility.
\emph{Econometrica}, 71(2):579--625, 2003.

\bibitem{abg12} A.~Arapostathis, V.~S.~Borkar, and M.~K.~Ghosh.
\emph{Ergodic Control of Diffusion Processes}. Encyclopedia of
Mathematics and its Applications 143, Cambridge University Press,
Cambridge, 2012.

\bibitem{asmussen03} S.~Asmussen.
\emph{Applied Probability and Queues}. 2nd edition, Springer, New York,
2003.

\bibitem{as08} M.~Avellaneda and S.~Stoikov.
High-frequency trading in a limit order book.
\emph{Quantitative Finance}, 8(3):217--224, 2008.

\bibitem{bacry13a} E.~Bacry, S.~Delattre, M.~Hoffmann, and J.-F.~Muzy.
Modelling microstructure noise with mutually exciting point processes.
\emph{Quantitative Finance}, 13(1):65--77, 2013.

\bibitem{bacry13b} E.~Bacry, S.~Delattre, M.~Hoffmann, and J.-F.~Muzy.
Some limit theorems for Hawkes processes and application to financial
statistics. \emph{Stochastic Processes and their Applications},
123(7):2475--2499, 2013.

\bibitem{bertram10} W.~K.~Bertram.
Analytic solutions for optimal statistical arbitrage trading.
\emph{Physica A}, 389(11):2234--2243, 2010.

\bibitem{bouchaud09} J.-P.~Bouchaud, J.~D.~Farmer, and F.~Lillo.
How markets slowly digest changes in supply and demand.
In \emph{Handbook of Financial Markets: Dynamics and Evolution}, Elsevier,
2009.

\bibitem{cjp15} \'A.~Cartea, S.~Jaimungal, and J.~Penalva.
\emph{Algorithmic and High-Frequency Trading}.
Cambridge University Press, 2015.



\bibitem{contdelarrard13} R.~Cont and A.~de~Larrard.
Price dynamics in a Markovian limit order market.
\emph{SIAM Journal on Financial Mathematics}, 4(1):1--25, 2013.

\bibitem{daleyverejones03} D.~J.~Daley and D.~Vere-Jones.
\emph{An Introduction to the Theory of Point Processes}. 2nd edition,
Springer, New York, 2003.

\bibitem{dayrirosenbaum15} K.~Dayri and M.~Rosenbaum.
Large tick assets: implicit spread and optimal tick size.
\emph{Market Microstructure and Liquidity}, 1(1):1550003, 2015.

\bibitem{delattre13} S.~Delattre, C.~Y.~Robert, and M.~Rosenbaum.
Estimating the efficient price from the order flow: a Brownian Cox
process approach. \emph{Stochastic Processes and their Applications},
123(7):2603--2619, 2013.

\bibitem{durbinkoopman12} J.~Durbin and S.~J.~Koopman.
\emph{Time Series Analysis by State Space Methods}. 2nd edition,
Oxford University Press, Oxford, 2012.



\bibitem{glf13} O.~Gu\'eant, C.-A.~Lehalle, and J.~Fernandez-Tapia.
Dealing with the inventory risk: a solution to the market making problem.
\emph{Mathematics and Financial Economics}, 7(4):477--507, 2013.

\bibitem{gueant16} O.~Gu\'eant.
\emph{The Financial Mathematics of Market Liquidity: From Optimal Execution
to Market Making}. Chapman \& Hall/CRC, 2016.

\bibitem{guilbaudpham13} F.~Guilbaud and H.~Pham.
Optimal high-frequency trading with limit and market orders.
\emph{Quantitative Finance}, 13(1):79--94, 2013.

\bibitem{hansenlunde06} P.~R.~Hansen and A.~Lunde.
Realized variance and market microstructure noise.
\emph{Journal of Business and Economic Statistics}, 24(2):127--161, 2006.

\bibitem{hasbrouck95} J.~Hasbrouck.
One security, many markets: determining the contributions to price
discovery. \emph{Journal of Finance}, 50(4):1175--1199, 1995.

\bibitem{hawkes71} A.~G.~Hawkes.
Spectra of some self-exciting and mutually exciting point processes.
\emph{Biometrika}, 58(1):83--90, 1971.

\bibitem{huanglehalle15} W.~Huang, C.-A.~Lehalle, and M.~Rosenbaum.
Simulating and analyzing order book data: the queue-reactive model.
\emph{Journal of the American Statistical Association},
110(509):107--122, 2015.

\bibitem{karlintaylor81} S.~Karlin and H.~M.~Taylor.
\emph{A Second Course in Stochastic Processes}. Academic Press, New York,
1981.

\bibitem{kramers40} H.~A.~Kramers. Brownian motion in a field of force
and the diffusion model of chemical reactions. \emph{Physica},
7(4):284--304, 1940.

\bibitem{leungli16} T.~Leung and X.~Li.
\emph{Optimal Mean Reversion Trading: Mathematical Analysis and Practical
Applications}. World Scientific, 2016.

\bibitem{liptonlopez20} A.~Lipton and M.~L\'opez~de~Prado.
A closed-form solution for optimal mean-reverting trading strategies.
Working paper, arXiv:2003.10502, 2020.

\bibitem{meyntweedie93} S.~P.~Meyn and R.~L.~Tweedie.
Stability of Markovian processes III: Foster--Lyapunov criteria for
continuous-time processes. \emph{Advances in Applied Probability},
25(3):518--548, 1993.

\bibitem{obizhaevawang13} A.~Obizhaeva and J.~Wang.
Optimal trading strategy and supply/demand dynamics.
\emph{Journal of Financial Markets}, 16(1):1--32, 2013.

\bibitem{pulido26} S.~Pulido, M.~Rosenbaum, and E.~Sfendourakis.
Understanding the worst-kept secret of high-frequency trading.
\emph{Finance and Stochastics}, 30(2):329--396, 2026.
doi:10.1007/s00780-026-00585-9.

\bibitem{rabechini19} L.~Rabechini Amaral and A.~Papanicolaou.
Price impact of large orders using Hawkes processes.
\emph{The ANZIAM Journal}, 61(2):161--194, 2019.

\bibitem{robertrosenbaum11} C.~Y.~Robert and M.~Rosenbaum.
A new approach for the dynamics of ultra-high-frequency data: the model
with uncertainty zones.
\emph{Journal of Financial Econometrics}, 9(2):344--366, 2011.

\bibitem{roll84} R.~Roll.
A simple implicit measure of the effective bid--ask spread in an efficient
market. \emph{Journal of Finance}, 39(4):1127--1139, 1984.

\bibitem{sfendourakis25} E.~Sfendourakis.
Multi-dimensional queue-reactive model and signal-driven models: a unified
framework. Working paper, arXiv:2506.11843, 2025.

\bibitem{zervos13} M.~Zervos, T.~C.~Johnson, and F.~Alazemi.
Buy-low and sell-high investment strategies.
\emph{Mathematical Finance}, 23(3):560--578, 2013.

\bibitem{zhangzhang08} H.~Zhang and Q.~Zhang.
Trading a mean-reverting asset: buy low and sell high.
\emph{Automatica}, 44(6):1511--1518, 2008.

\bibitem{zhangmykland05} L.~Zhang, P.~A.~Mykland, and Y.~A\"it-Sahalia.
A tale of two time scales: determining integrated volatility with noisy
high-frequency data. \emph{Journal of the American Statistical
Association}, 100(472):1394--1411, 2005.

\end{thebibliography}
\end{document}